\documentclass[showpacs,twocolumn,pra,superscriptaddress,notitlepage]{revtex4-2}
\usepackage[dvips]{graphicx}
\usepackage{amsmath,amssymb,amsthm,mathrsfs,amsfonts,dsfont}
\usepackage{subfigure, epsfig}
\usepackage{braket}
\usepackage{bm}
\usepackage{bbm}
\usepackage{enumerate}
\usepackage{color}
\usepackage{comment}
\usepackage{diagbox}
\usepackage{youngtab}
\usepackage[colorlinks = true]{hyperref}

\usepackage{enumitem}

\graphicspath{{./figure/}}

\newtheorem{theorem}{Theorem}

\newtheorem{prop}{Proposition}
\newtheorem{lemma}{Lemma}
\newtheorem{corollary}{Corollary}
\newtheorem{definition}{Definition}

\newcommand{\mc}{\mathcal}
\newcommand{\mb}{\mathbf}
\newcommand{\mbb}{\mathbb}
\newcommand{\tr}{\mathrm{Tr}}
\newcommand{\id}{\mathbb{I}}

\newcommand{\red}[1]{{\color{red} #1}}

\newcommand{\comments}[1]{}
\begin{document}
\title{Entanglement of random hypergraph states}

\begin{abstract}
Random quantum states and operations are of fundamental and practical interests. In this work, we investigate the entanglement properties of random hypergraph states, which generalize the notion of graph states by applying generalized controlled-phase gates on an initial reference product state. In particular, we study the two ensembles generated by random Controlled-Z(CZ) and Controlled-Controlled-Z(CCZ) gates, respectively. By applying tensor network representation and combinational counting, we analytically show that the average subsystem purity and entanglement entropy for the two ensembles feature the same volume law, but greatly differ in typicality, namely the purity fluctuation is small and universal for the CCZ ensemble while it is large for the CZ ensemble.  We discuss the implications of these results for the onset of entanglement complexity and quantum chaos.
\end{abstract}
 \date{\today}
\author{You Zhou}
\email{you\_zhou@fudan.edu.cn}
\affiliation{Key Laboratory for Information Science of Electromagnetic Waves (Ministry of Education), Fudan University, Shanghai 200433, China}
\affiliation{Centre for Quantum Technologies, National University of Singapore, 3 Science Drive 2, 117543 Singapore}
\affiliation{Nanyang Quantum Hub, School of Physical and Mathematical Sciences,Nanyang Technological University, Singapore 637371}
\author{Alioscia Hamma}
\email{alioscia.hamma@unina.it}
\affiliation{Physics Department, University of Massachusetts Boston, 02125, USA}
\affiliation{Dipartimento di Fisica Ettore Pancini, Universit\`a degli Studi di Napoli Federico II,
Via Cinthia, 80126 Napoli NA}
\maketitle
\section{Introduction}


Entanglement \cite{Horodecki2009entanglement} is the peculiar characteristic of quantum mechanics in multipartite systems, not only enabling quantum information processing with an advantage on its classical counterpart \cite{Nielsen2011Quantum}, but also finds fundamental applications in other branches of physics, such as condensed matter \cite{Amico2008Entanglement} and quantum gravity \cite{Qi2018gravity}.
Quantum hypergrah states \cite{Kruszynska2009Local,Rossi2013hyper,Qu2013hypergraph} are an archetypal class of multipartite states generalizing the notion of graph states \cite{Briegel2001Persistent,Hein2006Graph} by involving multi-qubit controlled-phase gates. Hypergraph state is of wide interests in 
measurement-based quantum computing \cite{miller2016hierarchy,Miller2018Latent,takeuchi2019quantum}, quantum advantage protocol \cite{Bremner2016IQP}, non-locality test \cite{guhne2014hypergraph,Gachechiladze2016Extreme}, quantum error correction and topological phase of matter \cite{Levin2012Braiding,Yoshida2016Topological,Miller2018Latent,Hamma2005Bipartite}. 


In this work, we study the entanglement statistics of ensembles of hypergraph states, specifically,
ensembles generated by random CZ and CCZ gates respectively.
By adopting a tensor network representation, we reduce the calculations to combinational problems. The main results of this paper are (i) the average subsystem purity and entanglement entropy in the two ensembles are essentially the same with Haar random states; on the other hand, (ii) Fluctuations for the subsystem purity are very different: while the CCZ ensemble features $O(d^{-2})$, Haar-like fluctuation, the CZ ensemble has larger $O(d^{-1})$ fluctuations with $d$ the total Hilbert space dimension. Consequently, for subsystem with half number of qubits, a volume law $S_2\simeq N/2-1$ for the entanglement entropy is typical for the CCZ ensemble with the variance of entanglement entropy exponentially small, scaling $\exp(-N)$ with the qubit number $N$, while the variance of the CZ ensemble is $O(1)$ and makes the expectation value not typical.


These results have a bearing on the different entanglement complexity of the hypergraph states generated by CZ and CCZ gates, and can inspire further studies on entanglement cooling algorithms \cite{Chamon2014Irreversibility,shaffer2014irreversibility,Yang2017localization}, quantum (pseudo-)randomness \cite{divincenzo2002quantum,dankert2009exact,collins2016random}, quantum magic \cite{veitch2012negative,veitch2014resource}, and quantum chaotic dynamics \cite{hayden2007black,roberts2017chaos} in many-body physics. In particular, random hypergraph states and generalized controlled-phase gates can provide toy models capturing the essential physics for complex dynamics while being easier to simulate. In addition, these techniques could be useful to study less structured quantum ensembles that are not necessarily an (approximate) $t$-design  \cite{ambainis2007quantum,gross2007evenly}. 

\section{Preliminaries}
\begin{figure}[hbt]
\centering
\resizebox{6cm}{!}{\includegraphics[scale=0.8]{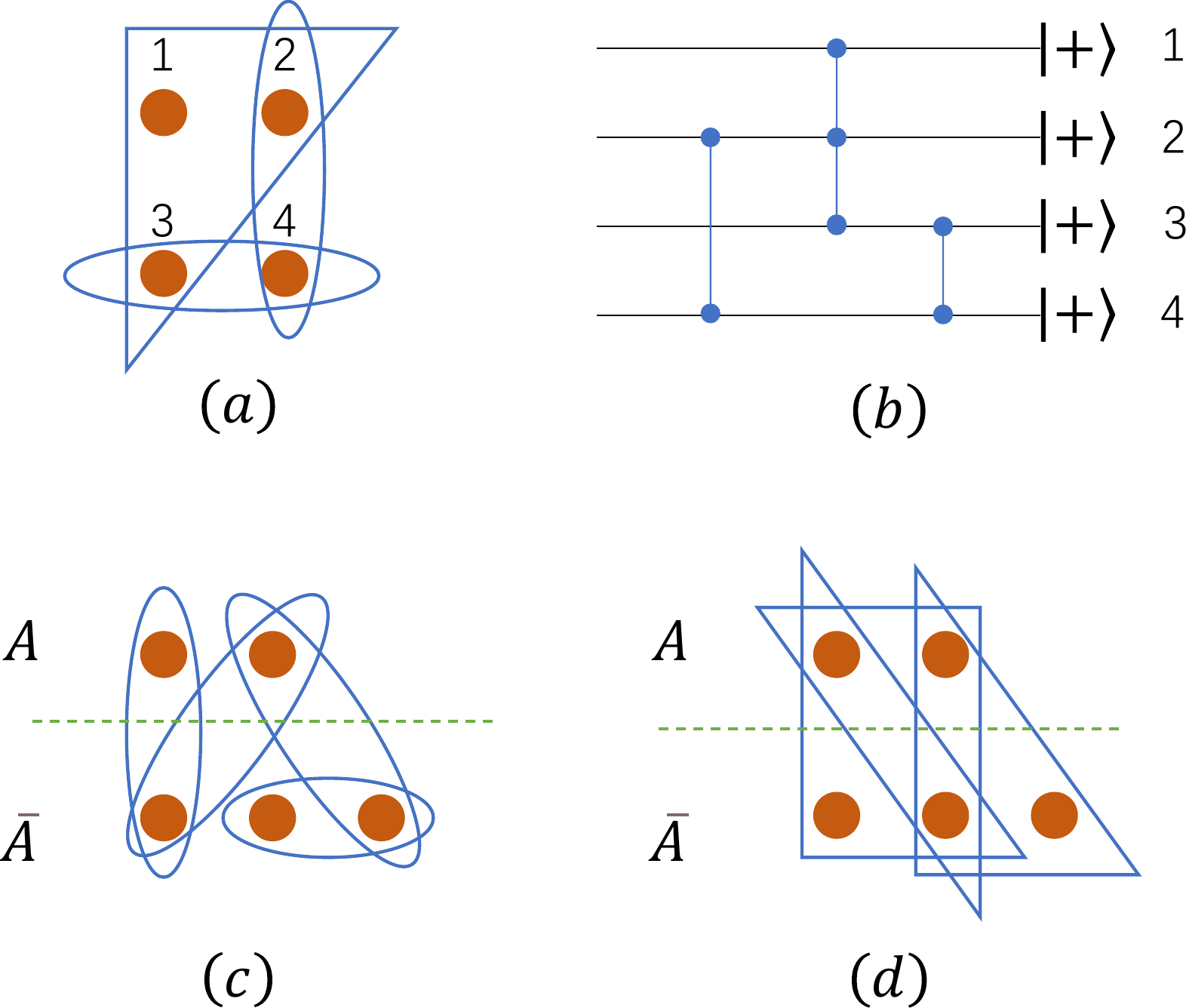}}
\caption{(a) A hypergraph with 4 vertices. There are three hyperedges, one of them containing the three vertices $1,2,3$. (b) The preparation quantum circuit in Eq.~\eqref{Eq:Gstate}, which evolves from \emph{right to left} hereafter. Vertical blue line denotes the $CZ_e$ gate, and the small circles on each qubit-line indicate the connection. In a 5-qubit system, (c) shows a 2-uniform hypergraph state, and (d) shows a 3-uniform one, respectively.}\label{Fig:Prelimi} 
\end{figure}

\subsection{Quantum hypergraph state}
First, let us recall the basic setup of  hypergraph states \cite{Rossi2013hyper,Qu2013hypergraph}. A hypergraph $G=(V,E)$, is a pair consisting of the vertex set $V=\{1,2\cdots,N\}$, and the edge set $E\subset V$.  If a hyperedge $e\in E$ contains  $1\leq k\leq N$ vertices, we say $e$ a $k$-edge. A hypergraph only containing $k$-edges is said $k$-uniform. See Fig.~\ref{Fig:Prelimi} for an illustration. We attach a local 
Hilbert space $\mathcal H_i$ at every vertex $i\in V$ and thus the total Hilbert space is $\mathcal H=\otimes_{i\in V}\mathcal H_i$. We are concerned with the case of qubits i.e., $\mathcal H_i\simeq\mathbb{C}^2$. The Hilbert space associated to an edge $e$ is $\mathcal H_e:= \otimes_{i\in e}\mathcal H_i$ and of course $\dim \mathcal H_e= 2^k$ for a $k$-edge. The total Hilbert space dimension is denoted as $d=2^N$.

To every hypergraph, we associate a (pure) quantum state in this way. First define the reference state  $\ket{\Psi_0}=\ket{+}^{\otimes N}$, with $\ket{+}=\frac1{\sqrt{2}}(\ket{0}+\ket{1})$ the eigenstate of Pauli $X$ operator. Then, the hypergraph state is obtained from the reference state by operating generalized controlled-Z gates according to the edges $e\in E$,
\begin{equation}\label{Eq:Gstate}
\begin{aligned}
\ket{G}=\left(\prod_{e\in E} CZ_e \right)\ket{\Psi_0}.
\end{aligned}
\end{equation}
Here the quantum gate $CZ_e=\otimes_{i\in e} \id_i-2\otimes_{i\in e} \ket{1}_i\bra{1}$ acts non-trivially on $\mc{H}_{e}$, but trivially on the qubits outside. See Fig.~\ref{Fig:Prelimi} (b) for the corresponding quantum circuit. For a 1-edge $e=\{i\}$, $CZ_e$ is just a Pauli Z gate on the qubit $i$; for a 2-edge, it becomes the normal Controlled-Z gate. Note that the sequence to operate the $CZ_e$ gate does not matter as they commute with each other, and thus the quantum hypergraph state $\ket{G}$ is uniquely determined by the hypergraph $G$. 

Previous works on the entanglement and correlation of hypergrah states  mainly deal with specific given states, see e.g. \cite{guhne2014hypergraph,Gachechiladze2016Extreme}, or results on entanglement witnesses  \cite{ghio2017multipartite}, local unitary transformations \cite{chen2014locally,gachechiladze2017graphical,tsimakuridze2017graph,lyons2015local}, and efficient verification \cite{morimae2017verification,zhu2019efficient}.
The structure of the states or the number of qubits in the system are usually restricted, since the scaling of the problems makes it quite challenging to extend. In this work, we instead consider the statistical results of entanglement properties of random hypergraph state ensembles (defined in the next section), and in the same time the behaviour under large qubit number $N$ can be analysed. We also remark that Ref.~\cite{collins2010random} studies random graph states, but follows a quite different definition.

\subsection{Random state ensembles}\label{sec:ensemble}

In this section, we define the random hypergraph state ensembles $\mc{E}$ from the corresponding random hypergraph ensembles. In particular, we consider random $2$-edge and random $3$-edge hypergraphs. A random $k$-uniform hypergraph can be determined by the probability $p$ whether there is a $k$-edge among any choice of $k$ sites. In the following, we focus on the case $p=1/2$.

According to Eq.~\eqref{Eq:Gstate}, the corresponding ensembles of hypergraph states are generated by CZ gates and CCZ gates, as shown in Fig.~\ref{Fig:Prelimi} (c)-(d). Denote the combination number $C^k_N:=\binom{N}{k}$, and formally the state ensemble is defined as follows.


\begin{definition}\label{Def:ensemble}
The (k-uniform) random hypergraph state ensemble $\mc{E}$ on $N$-qubit system is generated by
\begin{equation}\label{Eq:Def:ensemble}
\begin{aligned}
\mc{E}=\left\{\ket{\Psi}=U\ket{\Psi_0}\Big|U=U_{e_{C_{N}^k}}\cdots U_{e_2}U_{e_1}\right\}.
\end{aligned}
\end{equation}
Here each $e_i$ is a distinct $k$-edge of the $N$ vertices, with totally $C_{N}^k$ such edges,  $U_{e_i}$ acts on the Hilbert space $\mc{H}_{e_i}$ by taking $\{\id_{e_i},CZ_{e_i}\}$ with 1/2 probability, and $\ket{\Psi_0}=\ket{+}^{\otimes N}$.
\end{definition}
Since the $CZ_{e}$ gates commute, the order of the gate sequence in Eq.~\eqref{Eq:Def:ensemble} is irrelevant. Of course, there are $2^{C_{N}^{k}}$ elements in $\mc{E}$ with equal probability, that is, $|\mc{E}|=2^{C_{N}^k}$. In this work, we focus on the $k=2$ and $k=3$ random hypergraph state ensembles, denoted by $\mc{E}_{CZ}$ and $\mc{E}_{CZZ}$, respectively.   Of course the $k=1$ case is trivial as this is just a single qubit gate, so that the corresponding graph states are just product states. The first non-trivial gate is the $k=2$, CZ gate. Note that CZ is a Clifford gate, but other $CZ_{e}$ for $k>2$ are not \cite{gottesman1998heisenberg}. By adding the Hadamard  gate to $CZ_{e}$, the gate set becomes universal for quantum computing \cite{shi2002both}. And we conjecture that the results
of CCZ hold for other $CZ_e$ with some constant $k>3$. 

\subsection{Entanglement entropy and purity}
Consider a bipartition of the $N$-qubit system into $\{A,\bar{A}\}$ with $N_A$ and $N_{\bar{A}}$ qubits respectivley, and the total Hilbert space $\mathcal H =\mathcal H_A\otimes \mathcal H_{\bar{A}}$.
The bipartite entanglement with respective to $\{A,\bar{A}\}$ of a pure quantum state $\ket{\Psi}$ is quantified by the R\'enyi-$\alpha$ entanglement entropy defined as
\begin{equation}\label{Eq:EntDef}
\begin{aligned}
S_{\alpha}(\rho_A)=\frac{1}{1-\alpha}\log_2[\mathrm{Tr}(\rho_A^{\alpha})],
\end{aligned}
\end{equation}
where $S_{\alpha}$ is the the R\'enyi-$\alpha$ entropy of the reduced density matrix $\rho_A=\mathrm{Tr}_{\bar{A}}(\Psi)$, with $\Psi=\ket{\Psi}\bra{\Psi}$. As $\alpha=1$, Eq.~\eqref{Eq:EntDef} gives the Von Neumann entropy. The R\'enyi-2 entropy shows
\begin{equation}
\begin{aligned}
S_2(\rho_A)=-\log_2(P_A),
\end{aligned}
\end{equation}
with $P_A=\mathrm{Tr}(\rho_A^2)$ the purity functional. The
R\'enyi-2 entropy is a lower bound of the Von Neumann entropy. This quantity is particularly useful as the purity $P_A$ can be directly measured in experiment \cite{Islam2015Measuring,Brydges2019Probing}. 

In the next sections, we compute the average subsystem purity and the variance of the purity for the state $\Psi$ from the ensembles $\mc{E}=\mc{E}_{CZ}, \mc{E}_{CCZ}$, that is,
\begin{equation}\label{Eq:FlucDef}
\begin{aligned}
\langle P_A \rangle_{\mc{E}}&:=\mbb{E}_{ \mc{E}}(P_{A}),\\
\delta^2_\mc{E}(P_{A})&:=\mbb{E}_{\mc{E}}(P_{A}^2)-[\mbb{E}_{\mc{E}}(P_{A})]^2
\end{aligned}
\end{equation}
Hereafter we also use $\langle \cdot \rangle_{\mc{E}}$ to denote the ensemble average.

\section{Average subsystem purity}\label{sec:avrPur}
In this section, we study the average purity of a subsystem $A$ for random hypergraph states. By introducing a tensor network representation in Sec.~\ref{sec:TNrep} to facilitate the calculation, the problem is transformed to some combinational counting, under different constraints induced by the CZ and CCZ operations, respectively. And this methodology is also applied to the variance of purity in Sec.~\ref{sec:variance}.

The average purity of the state $\rho_A=\tr_{\bar{A}}\Psi$ with $\Psi\in\mc{E}$ can be expressed, by standard technique (e.g. \cite{Hamma2012Random}), as
\begin{equation}\label{Eq:Pur}
\begin{aligned}
\mbb{E}_{\Psi}(P_{A})&=\mbb{E}_{\Psi}\mathrm{Tr}\left(T_A\otimes \id_{\bar{A}}^{\otimes 2}\ \Psi^{\otimes 2}\right)\\
&=\mbb{E}_{U\in\mc{E}} \mathrm{Tr}\left(T_A\otimes \id_{\bar{A}}^{\otimes 2} \ U^{\otimes 2}\Psi_0^{\otimes 2} U^{\dag\otimes 2}\right)\\
&=\mathrm{Tr}\left[T_A\otimes \id_{\bar{A}}^{\otimes 2} \ \mbb{E}_{U\in\mc{E}}\left(U^{\otimes2}\Psi_0^{\otimes 2} U^{\dag\otimes 2}\right)\right]\\
&=\mathrm{Tr}\left[T_A\otimes \id_{\bar{A}}^{\otimes 2}\ \Phi_\mc{E}^2\left(\Psi_0^{\otimes 2}\right)\right].
\end{aligned}
\end{equation}
Above, in the first line the purity is written as an observable on the two-copy Hilbert space $\mc{H}^{\otimes 2}$, and $T_A$ is the swap operator on $\mc{H}_A^{\otimes 2}$. The random state $\Psi$ is generated by random unitary $U$ on the initial state $\Psi_0$, and we still denote the unitary ensemble with $\mc{E}$ without ambiguity. In the third line of Eq.~\eqref{Eq:Pur},  due to the linearity, the order of trace and ensemble average are exchanged, and we denote the resulting 2-copy``twirling" channel as $\Phi_\mc{E}^2$ in the final line. See Fig.~\ref{Fig:2qPur} (a) for an illustration.

To study the random hypergraph state, the initial state is taken as $\ket{\Psi_0}=\ket{+}^{\otimes N}$, with the subsystem $A$ containing $N_A\le N/2\le N_{\bar{A}}$ qubits without loss of generality. 
The evolution $U$ is sampled from CZ or CCZ ensembles defined in Sec.~\ref{sec:ensemble}. 
Actually, one only needs to consider the gate between $A$ and $\bar{A}$, since the gate inside each subsystem commutes with $T_A$ and thus does not affect the purity.

\comments{
By using the cycle property of trace, we can also equivalently put the ``twirling'' channel on $T_A$,
\begin{equation}\label{Eq:Pur1}
\begin{aligned}
\mbb{E}_{\Psi}(P_{A})
&=\mathrm{Tr}\left[\mbb{E}_{U\in\mc{E}} U^{\dag\otimes 2}(T_A\otimes \id_{\bar{A}}^{\otimes 2}) U^{\otimes 2}\ \ket{+}\bra{+}^{\otimes 2N} \right].
\end{aligned}
\end{equation}
}

\subsection{Tensor network representation}\label{sec:TNrep}
Here we introduce a tensor network representation to calculate the average purity. To make the tensor diagram concise, we take the quantum state $\rho$ and observable $O$ as \emph{vectors}, and the quantum channel $\Phi$ as a \emph{matrix}. That is, in the form of $\langle\langle O|\Phi|\rho\rangle\rangle$. In our case, $O=T_A\otimes \id_{\bar{A}}^{\otimes 2}$, $\rho=\ket{+}\bra{+}^{\otimes 2N}$, and $\Phi$ is the twirling channel $\Phi_\mc{E}^2$, as shown in the final line of Eq.~\eqref{Eq:Pur}.
In particular, one has
\begin{equation}\label{Eq:PurVec}
\begin{aligned}
\mbb{E}_{\Psi}(P_{A})=
\mbb{E}_{U\in\mc{E}} \langle\langle T_A|\ \otimes \langle\langle\id_{\bar{A}}^{2}| U\otimes U^*\otimes U\otimes U^*\ \ket{\Psi_0}^{\otimes 4},
\end{aligned}
\end{equation}
which is illustrated in Fig.~\ref{Fig:2qPur} (b). There are two different connections on the left, which correspond to the vector form of the operators $\langle\langle T_A|$ and $\langle\langle\id_{\bar{A}}^{2}| $ in Eq.~\eqref{Eq:PurVec}, respectively.

For the random hypergraph states here, the random unitary satisfies $U=U^*$, and the initial state $\ket{\Psi_0}$ is real and product state. In Fig.~\ref{Fig:2qPur} (c), we further rearrange the total Hilbert space $\mc{H}^{\otimes 2} $ into $(\mc{H}_i^{\otimes 2} )^{\otimes N}$, i.e., put the two-copy of the i-th qubit Hilbert space $\mc{H}_i$ together. On the right, every qubit corresponds to 4 lines, each 2 for one-copy Hilbert space. In the middle, we operate the random $CZ_{e}$ gate across $A$ and $\bar{A}$. Here we show case of a CZ gate, which repeats 4 times connecting the corresponding lines of the two-qubit in the tensor network. Note that the swap $T_A=\otimes_{i\in A} T_i$, and of course the identity operator $\id_{\bar{A}}^{\otimes 2}$ can be written in the product of each qubit operators. As a result, on the left, there are two different connections for each qubit determined by it belonging to $A$ or $\bar{A}$.


\begin{figure}[hbt]
\centering
\resizebox{9.5cm}{!}{\includegraphics[scale=0.8]{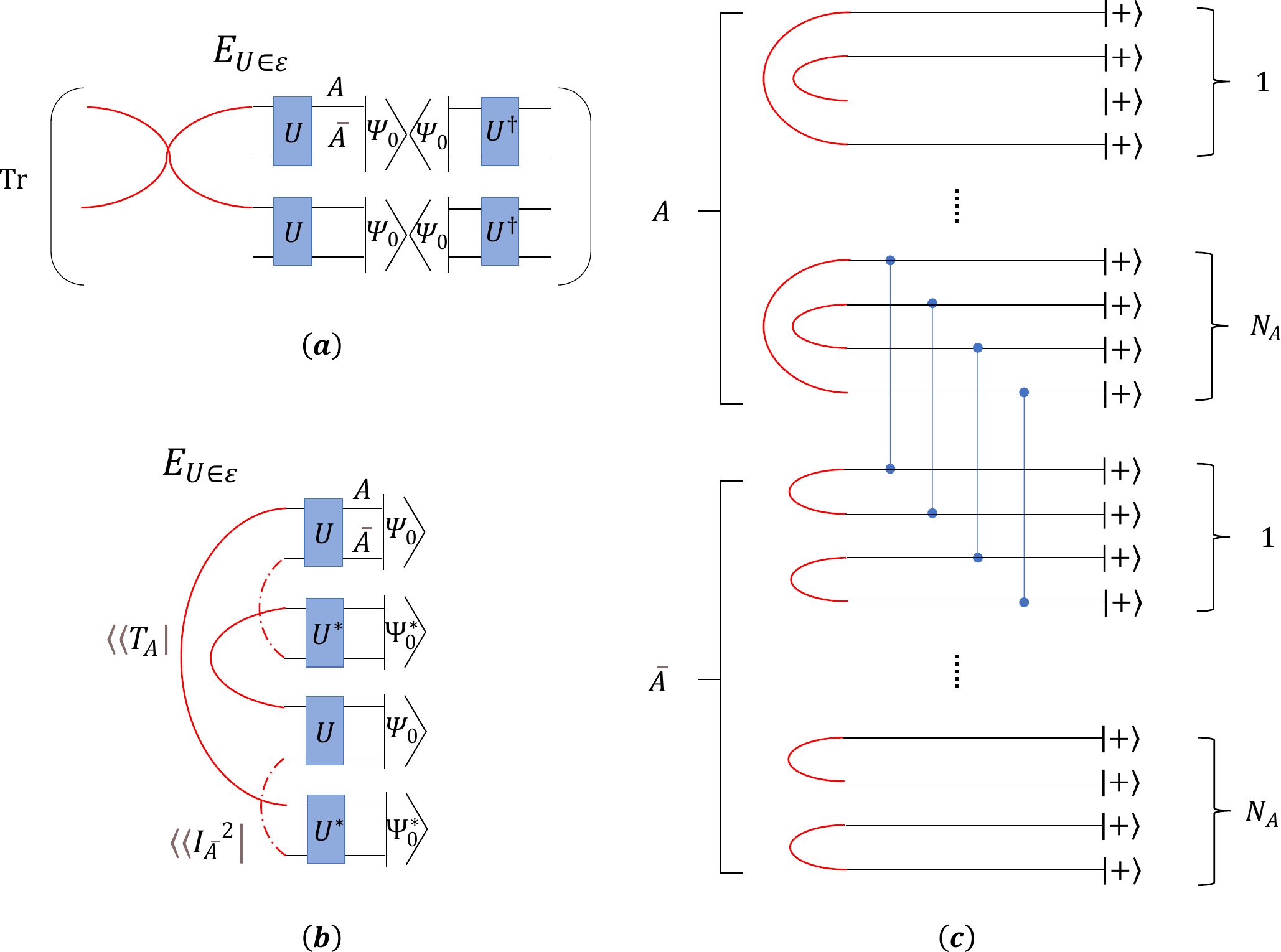}}
\caption{Tensor network of the purity formula and its vectorization. 
(a) The purity formula of Eq.~\eqref{Eq:Pur}, with the red ``cross'' being the swap operator on subsystem $A$. (b)  Rearrangement of the tensors, with swap and identity operators becoming two kinds of `vector' on the left. (c) Take the input state as $\ket{\Psi_0}=\ket{+}^{\otimes N}$ and further rearrange the tensor network to put the two-copy of the i-th qubit Hilbert space $\mc{H}_i$ together. 
As a result, on the right, every qubit corresponds to 4 lines, each 2 for one-copy Hilbert space; in the middle, we operate the random $CZ_{e}$ gate across $A$ and $\bar{A}$. An example of CZ gate is shown, repeating 4 times connecting the corresponding lines of the two-qubit; on the left, the original swap and identity in (b) are both decomposed to the qubit-level, representing two different connections for qubit in $A$ and $\bar{A}$. }\label{Fig:2qPur}
\end{figure}

To calculate $\mbb{E}_{\Psi}(P_{A})$, one needs to contract the tensor network in Fig.~\ref{Fig:2qPur} (c) for all possible random $U$ and then average them.
Note that in Fig.~\ref{Fig:2qPur} (c), the input vector $(\ket{+}^{\otimes 4} )^{\otimes N}$ on the right can take all the possible $0/1$ bit values with a normalization constant

\begin{equation}\label{eq:NormalPur}
\begin{aligned}
\mathcal{N}=(\frac1{\sqrt{2}})^{4N}=d^{-2}.
\end{aligned}
\end{equation}

The unitary in the middle is composed of $CZ_{e}$ gates, and only introduce $\pm 1$ phase for any computational basis input.
By averaging on these random unitaries, some bit-strings may be cancelled out by taking $-1$ and $1$ for different unitaries. As a result, to find the final average purity one only needs to count the number of bit-strings which ``survive'' on the average effect of all possible gate operations. The following proposition can simplify the counting procedure by observing that the twirling channel in the middle can be decomposed \emph{locally}, since the gate operations on distinct edges are \emph{independent}.


\begin{prop}\label{prop:ChannelLocal}
Consider a $t$-fold twirling channel $\Phi_\mc{E}^t(\cdot):=\mbb{E}_{U\in\mc{E}} U^{\otimes t}(\cdot) U^{\dag\otimes t}$, with random unitary $U=U_{e_L}\cdots U_{e_2}U_{e_1}$, where each $U_{e_i}$ is sampled independently from the local ensemble on the Hilbert space $\mc{H}_{e_i}$. $\Phi_\mc{E}^t(\cdot)$  can be decomposed into the channel multiplication of the corresponding local twirling channels
\begin{equation}\label{eq:decompose}
\begin{aligned}
\Phi_\mc{E}^t=\Phi_{\mc{E}_{e_L}}^t \circ \Phi_{\mc{E}_{e_{L-1}}}^t \cdots \circ \Phi_{\mc{E}_{e_2}}^t\circ \Phi_{\mc{E}_{e_1}}^t
\end{aligned}
\end{equation}
where $\mc{E}_{e_i}$ denotes the random unitary ensemble acting non-trivially on $\mc{H}_{e_i}$.
\end{prop}
\begin{proof}
We prove it by using the matrix form of the twirling channel $\tilde{\Phi}_\mc{E}^t=\mbb{E}_{U\in\mc{E}} U^{\otimes t}\otimes U^{*\otimes t}$, where the tilde denotes the matrix form, which is also called the tensor product expander \cite{low2010pseudo,Toth2007Efficient}. In this way, the composition of channel can be taken as the multiplication of the matrix,
\begin{equation}\label{}
\begin{aligned}
\tilde{\Phi}_\mc{E}^t&=\mbb{E}_{U=U_{e_L}\cdots U_{e_2}U_{e_1}} [U_{e_L}\cdots U_{e_2}U_{e_1}]^{\otimes t}\otimes [U_{e_L}\cdots U_{e_2}U_{e_1}]^{*\otimes t}\\
&=[\mbb{E}_{U_{e_L}\in\mc{E}_{e_L}}U_{e_L}^{\otimes t}\otimes U_{e_L}^{*\otimes t}]\cdots[\mbb{E}_{U_{e_1}\in\mc{E}_{e_1}}U_{e_1}^{\otimes t}\otimes U_{e_1}^{*\otimes t}]\\
&=\tilde{\Phi}_{\mc{E}_{e_L}^t} \cdots \tilde{\Phi}_{\mc{E}_{e_2}^t}\cdot\tilde{\Phi}_{\mc{E}_{e_1}^t},\\
\end{aligned}
\end{equation}
where the second line we use the independence of each unitary ensemble $\mc{E}_{e_i}$,
then we translate this back to the channel form and get Eq.~\eqref{eq:decompose}.
\end{proof}

Prop.~\ref{prop:ChannelLocal} is general and suitable for any kind of random unitary ensemble. For our random $CZ_e$ ensemble in Def.~\ref{Def:ensemble}, each $e_i$ denotes a specific edge, and $U_{e_i}\in \{\id_{e_i},CZ_{e_i}\}$ acting nontrivially on $\mc{H}_{e_i}$ with 1/2 probability, which gives the local twirling channel $\Phi_{\mc{E}_{e_i}}^t$.  Note that the phase gates commutes with each other, and thus the order of the decomposition in Eq.~\eqref{eq:decompose} does not matter.


\subsection{Average purity of CZ ensemble}\label{CZpurity}
\begin{figure}[hbt]
\centering
\resizebox{9cm}{!}{\includegraphics[scale=0.8]{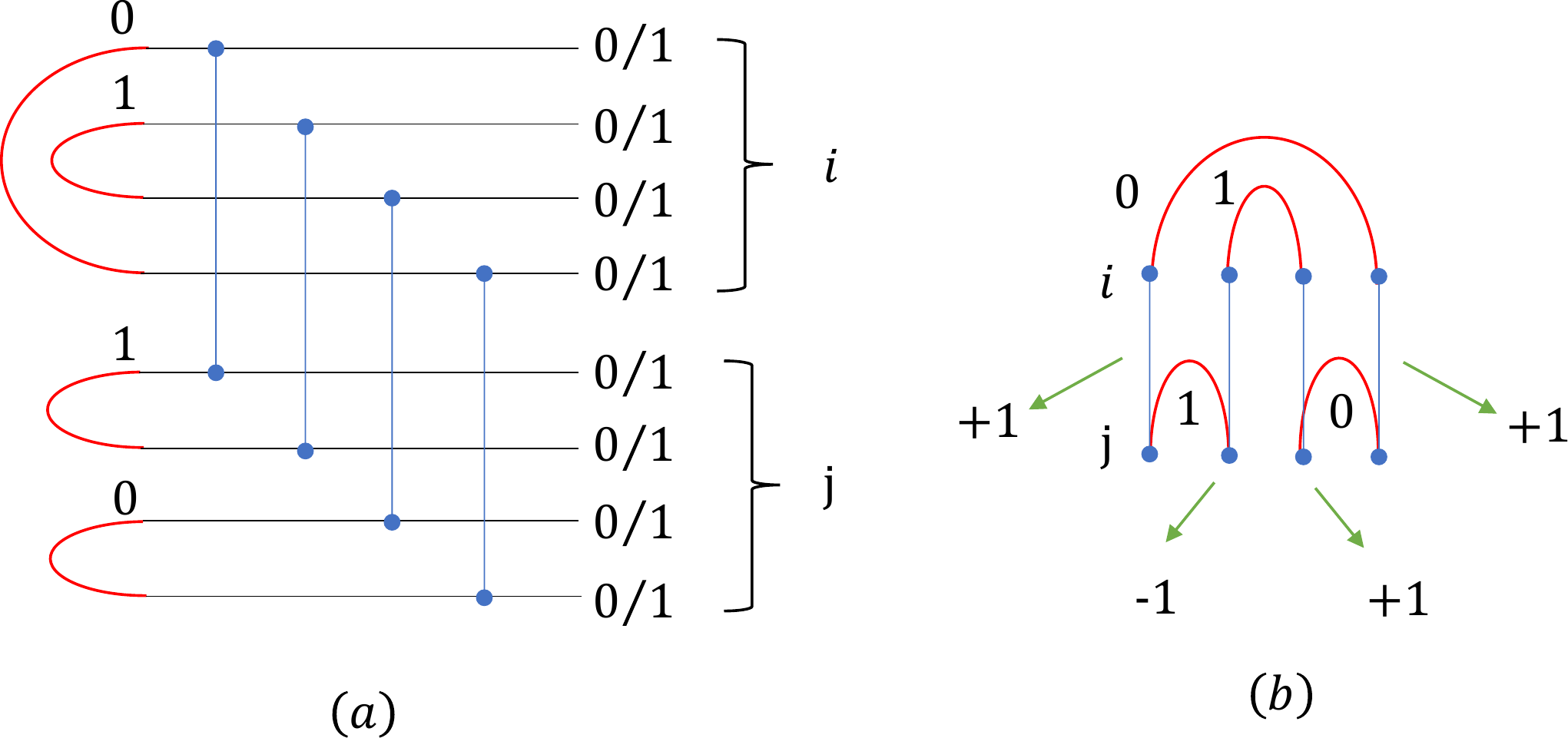}}
\caption{(a) Tensor network for the two-qubit CZ operation. On the right, there are all possible $0/1$ inputs of all black lines. As black lines associated to the same red arc should take the same bit value, one only needs to use the 2-bit of the red arcs to denote each qubit, and here we show case for  $\{01,10\}$. (b)  We rotate (a) 90 degrees for clearness, and also omit the nonessential black lines. The qubit $i$ and $j$ are labeled by the the upper and lower two-red-arcs, with every arc can taking $0/1$ value. The overall phase is determined by the phase accumulation of all four CZ gates. For the $\{01,10\}$ case here, the overall phase is $-1=(+1)*1*(-1)*(+1)*(+1)$. The corresponding phase matrix and the corresponding diagram.
}\label{Fig:PurCZ}
\end{figure}
Equipped with the tensor network diagram, we are in the position to calculate the average purity for CZ ensemble.
To facilitate the counting procedure, we should first figure out what happens locally, say the twirling channel $\Phi^2_{\mc{E}_{e}}$ on a specific edge $e=\{i,j\}$.

First, for simplicity, let us consider a 2-qubit example, and then extend to N-qubit later. Suppose the qubit $i$ is in $A$ and $j$ is in $\bar{A}$, as shown in Fig.~\ref{Fig:PurCZ} (a). Remember that the input state $\ket{+}$ on the right induces no constraint, i.e., every black line can take $0/1$ value by ignoring the normalization $\mathcal{N}$ , thus every qubit can be denoted by 4 \emph{classical} bits. 
Moreover, since the red tensors on the left connect the ends of the black lines, they induce constraints on the bit values. To be specific, any two black lines should take the same bit value if they are connected to the same red arc. Otherwise it will return zero in the tensor contraction, and thus does not contribute to the summation, no matter what kind of gate operations in the middle.  In other words, the bit values of a qubit can be reduced to that of the red arcs, and consequently one can use the red arcs on the left to label the qubit and only needs to associate 2-bit to each qubit. Fig.~\ref{Fig:PurCZ} (a) shows the case for $\{01,10\}$.

In Fig.~\ref{Fig:PurCZ} (b), for the simplicity of the following discussion, we rotate the diagram 90 degrees and omit the nonessential black lines. The CZ gate on the 2-qubit can introduce phase, depending on the value of the classical bits of the red arcs. To be specific, if there are both 1s of the red arcs on the two ends of the CZ gate, it will give a $-1$ phase, otherwise the phase is 1, and the final phase is determined by mutiplication of all phases from the four CZ gates. For the $\{01,10\}$ case, the overall phase is $-1$ as shown in Fig.~\ref{Fig:PurCZ} (b).
We list all possible cases as follows
\begin{equation}\label{}
M_p=\bordermatrix{%
&00&01&10&11\cr
00&1 & 1 & 1 &1\cr
01&1& -1& -1 & 1\cr
10&1 & -1 & -1 &1\cr
11&1& 1 & 1 & 1\cr
},
\end{equation}
and denote it as the \emph{phase} matrix, where the row and column indexes are the 2-bit of the two qubits respectively. 
The phase matrix is symmetric between the two qubits, also invariant under permutation of the 2-bit for the the row or column, and these properties are also reflected in the diagram of Fig.~\ref{Fig:PurCZ} (b).

On the other hand, if there is no gate operation, the phase matrix is just the trivial one
\begin{equation}\label{}
J_4=\bordermatrix{%
&00&01&10&11\cr
00&1 & 1 & 1 &1\cr
01&1& 1& 1 & 1\cr
10&1 & 1 & 1 &1\cr
11&1& 1 & 1 & 1\cr
}.
\end{equation}

After taking the average for both cases, the final matrix is the summation of them with $1/2$ probability as
\begin{equation}\label{eq:sumMfull}
M_s=\bordermatrix{%
&00&01&10&11\cr
00&1 & 1 & 1 &1\cr
01&1& 0& 0 & 1\cr
10&1 & 0 & 0 &1\cr
11&1& 1 & 1 & 1\cr
},
\end{equation}
denoted as the \emph{summation} matrix.

$M_s$ in some sense can be regarded as a matrix representation of the local twirling channel $\Phi^2_{\mc{E}_{e}}$, and we elaborate this point in App.~\ref{ap:sumM&TrChannel}. In the remaining of the paper, we will adopt this summation matrix formalism to calculate the average purity and its variance.

The matrix $M_s$ in Eq.~\eqref{eq:sumMfull} indicates that only a few of bit configurations survive under the effect of the local twirling channel $\Phi^2_{\mc{E}_{e}}$. For instance, the configuration $\{00,11\}$ survives, but the configuration $\{01,10\}$ vanishes. To figure the average purity in this minimal $N=2$ case, one only needs to count the number of $1$s in Eq.~\eqref{eq:sumMfull}, i.e, $12$, and normalizes it with $\mathcal{N}=1/16$ in Eq.~\eqref{eq:NormalPur} to get $3/4$. This is consistent with the direct calculation: there is $1/2$ probability to prepare a product state with purity $1$ on $A$, and the other $1/2$ probability to prepare the Bell state with purity $1/2$, which totally returns the average $(1+1/2)/2=3/4$.

This counting argument can be extended to general N-qubit system by virtue of Prop.~\ref{prop:ChannelLocal}. We still use 2-bit to denote one qubit and our task is to count the total number of bit-strings which survive under all possible local twirling channels $\Phi^2_{\mc{E}_{e}}$. Since there is $\Phi^2_{\mc{E}_{e}}$ between any pair of qubits, we just need to move the summation matrix $M_s$ in Eq.~\eqref{eq:sumMfull} on all 2-qubit between $A$ and $\bar{A}$.
For example, if the two-qubit say $i,j$ with $i\in A$ and $j\in \bar{A}$, take the bit configuration such as $\{01,10\}$, no matter what values of other qubits are, this kind of bit-string $\{\cdots01_{\{i\}},10_{\{j\}}\cdots\}$ will vanish and not contribute to the final result.

Note that the matrix elements just depend on the parity information, one can simplify it further by using logical bit $\bar{0}=00,11$ and $\bar{1}=01,10$ to encode such that
\begin{equation}\label{Eq:sumMCZ}
M_s=\bordermatrix{%
&\bar{0}&\bar{1}\cr
\bar{0}&1 & 1 \cr
\bar{1}&1& 0\cr
}.
\end{equation}
and the phase matrix shows
\begin{equation}\label{Eq:phMCZ}
M_p=\bordermatrix{%
&\bar{0}&\bar{1}\cr
\bar{0}&1 & 1 \cr
\bar{1}&1& -1\cr
}.
\end{equation}

Now every qubit is labeled by \emph{one} classical bit, and we utilize $M_s$ to constrain the bit value of any qubit-pair in $A$ and $\bar{A}$.
From the matrix $M_s$, one sees that $\bar{1}$ is not allowed 
in $A$ and $\bar{A}$ simultaneously. That is, all qubits in $A$ take $\bar{0}$ and those in $\bar{A}$ are arbitrary or vice versa. Consequently, the total number of surviving bit-string for the random CZ scenario is
\begin{equation}
\begin{aligned}
\#_{CZ}=\left(2^{N_A}+2^{N_{\bar{A}}} -1\right)2^{N}.
\end{aligned}
\end{equation}
Here $-1$ accounts for deleting the double counting of the all $\bar{0}$ case; the multiplication $2^{N}$ is due to the redundancy of the logical encoding. The average purity is obtained by multiplying $\#_{CZ}$ with the normalization $\mathcal{N}$ in Eq.~\eqref{eq:NormalPur}.
\begin{theorem}\label{th:CZPur}
The average purity of the subsystem $A$ with $N_A$ qubits for the random graph states from the CZ ensemble is
\begin{equation}\label{eq:avrPurCZ}
\begin{aligned}
\langle P_A \rangle_{CZ}
=\frac{d_A+d_{\bar{A}}-1}{d}
\end{aligned}
\end{equation}
where $d_{A(\bar{A})}=2^{N_{A(\bar{A})}}$ is the Hilbert space dimension of the subsystem $A(\bar{A})$. For the case of equal partition $d_A=d_{\bar{A}}=\sqrt{d}$, one has $\langle P_A \rangle_{CZ}=(2\sqrt{d}-1)/d \sim 2/\sqrt{d}$.
\end{theorem}
\begin{figure}[hbt]
\centering
\resizebox{6cm}{!}{\includegraphics[scale=0.8]{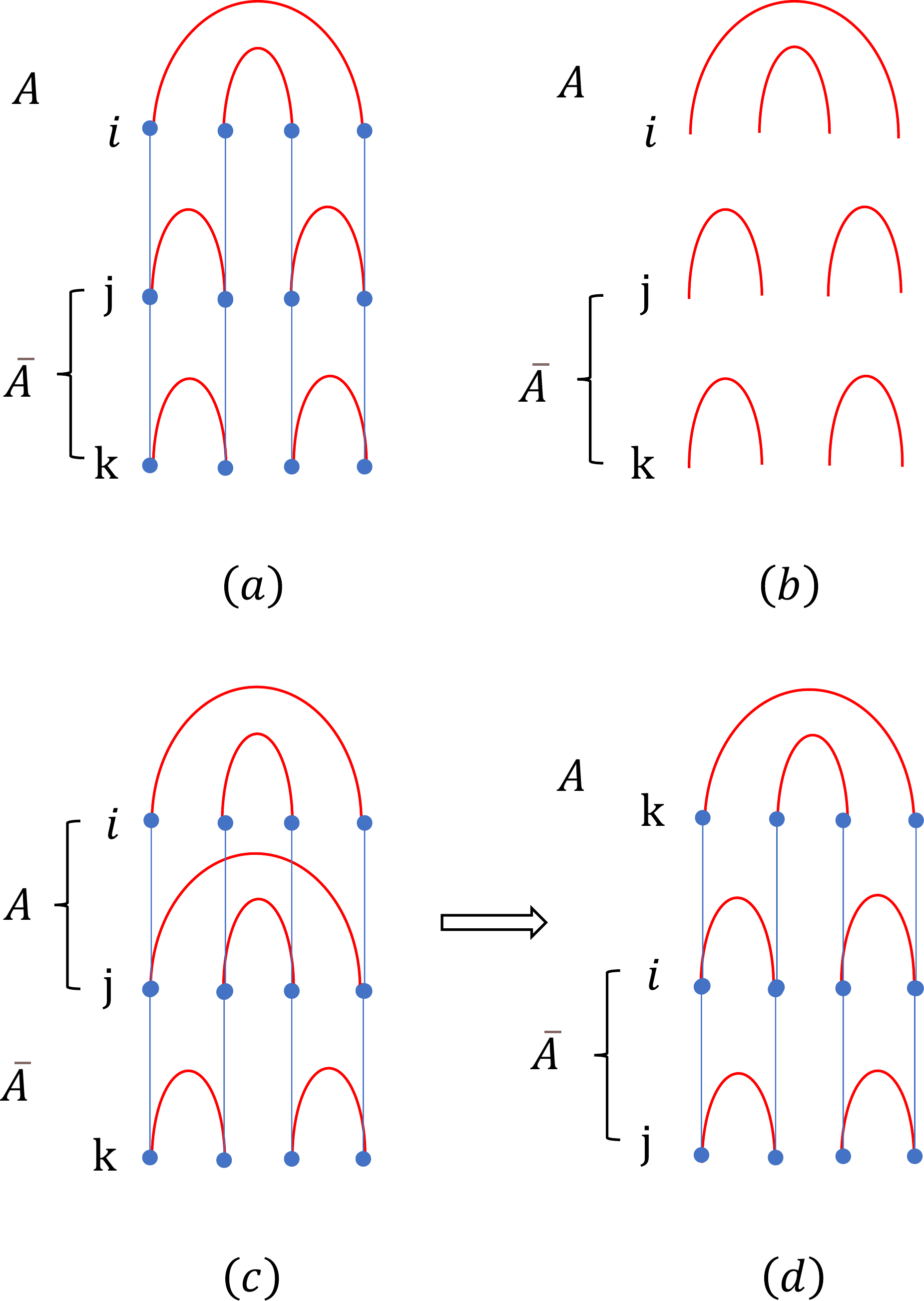}}
\caption{Phase matrix (3-index tensor) for the CCZ scenario. For (a) and (b), qubit $i$ is in $A$, the other two qubits $j,k$ are in $\bar{A}$. (a) is the phase tensor $M_{ijk}^p$ when operating the CCZ gate, and (b) is the trivial all-ones tensor $J_{ijk}$ without gate operation. The summation tensor $M_{ijk}^s$ is the average of both. (c) is the phase tensor $\tilde{M}_{ijk}^p$ for the case of $i,j \in A,k\in \bar{A}$, which is different from $M_{ijk}^p$ due to the shape of the red tensors. But it can be deformed to $M_{kij}^p$ in (d) without changing the topology.}\label{Fig:PurCCZ}
\end{figure}

\subsection{Average purity of CCZ ensemble}\label{CCZpurity}
Now we move to the scenario of random CCZ hypergraph states. Similarly with the CZ scenario in Sec.~\ref{CZpurity}, we still use two classical bits to denote one qubit, and count the number of the surviving bit-strings.

First, we should figure out what happens locally, say the corresponding summation matrix $M_s$ caused by the local twirling channel $\Phi^2_{\mc{E}_{e}}$.
The only difference compared with the CZ scenario is that CCZ involves three qubits $i,j,k$, and thus the corresponding phase matrix $M_p$ is actually a \emph{three-index tensor} $M_{ijk}^p$, and we denote the trivial all-ones tensor without gate operation as $J_{ijk}$, as shown in Fig.~\ref{Fig:PurCCZ} (a) and (b). The summation tensor is their average
\begin{equation}
\begin{aligned}
M_{ijk}^s=\frac1{2}\left(M_{ijk}^p+J_{ijk}\right).
\end{aligned}
\end{equation}

\begin{figure}[hbt]
\centering
\resizebox{8cm}{!}{\includegraphics[scale=0.8]{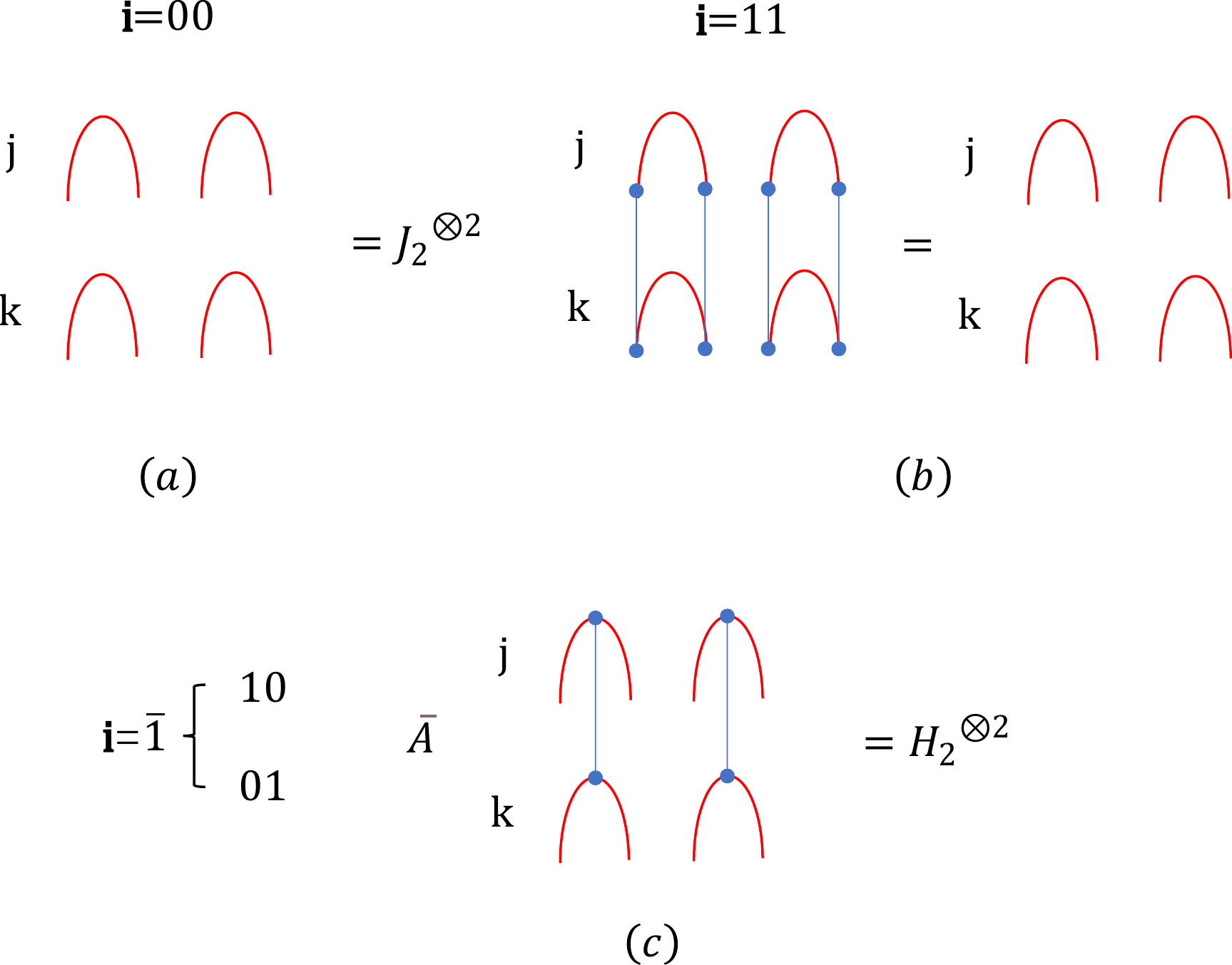}}
\caption{Phase tensor $M_{ijk}^p$ in Fig.~\ref{Fig:PurCCZ} (a) for different classical bit values of qubit $i$. (a) For $\mb{i}=00$, the blue CCZ gate can not introduce phase on the $j,k$ qubits, and thus the matrix between $j,k$ is $J_2^{\otimes 2}$. (b) For $\mb{i}=11$, the CCZ gates reduce to CZ gates on $j,k$. However, the nearby two CZ gates cancel with each other and lead also to $J_2^{\otimes 2}$. (c) For $\mb{i}=01,10$, there is only one CZ gate on a pair of red arcs, and the corresponding matrix form is $H_2^{\otimes 2}$.}\label{Fig:PurCCZ1}
\end{figure}
Suppose qubit $i\in A$, and the other two qubits $j,k\in \bar{A}$ for the case in Fig.~\ref{Fig:PurCCZ} (a), we can write down $M_{ijk}^p$ by fixing the bit value of $i$. 
If qubit $i$ takes the logical $\bar{0}$, denoted by $\mb{i}=\bar{0}$ with the blod font, the matrix between $j$ and $k$ is the trivial $4\times4$ all-ones matrix $J_4=J_2^{\otimes 2}$ in Fig.~\ref{Fig:PurCCZ1} (a) and (b). In this case, the summation tensor is
\begin{equation}\label{eq:CCZtrivial}
\begin{aligned}
M_{ijk}^s&=\ket{\mb{i}}\bra{\mb{i}}\otimes \frac{1}{2}(J_4+J_4)\\
&=\ket{\mb{i}}\bra{\mb{i}}\otimes J_4.
\end{aligned}
\end{equation}
with $\mb{i}=00,11$ and Dirac notation is adopted to write the tensor.

\comments{
\begin{equation}\label{Eq:phMCCZ}
H^{\otimes 2}=\bordermatrix{%
&00&01&10&11\cr
00&1 & 1 & 1 &1\cr
01&1& -1& 1 & -1\cr
10&1 & 1 & -1 &-1\cr
11&1& -1 & -1 & 1\cr
}.
\end{equation}
}
If qubit $\mb{i}=\bar{1}$, the matrix between $j,k$ is $H_2^{\otimes 2}$
in Fig.~\ref{Fig:PurCCZ} (c), with $H_2$ being the $2\times 2$ Hadamard matrix \emph{without} normalization constant $1/\sqrt{2}$ hereafter. And the resulting summation tensor is
\begin{equation}\label{eq:MsPurCCZ}
\begin{aligned}
M^s_{ijk}&=\ket{\mb{i}}\bra{\mb{i}}\otimes 1/2(H_2^{\otimes 2}+J_2^{\otimes 2})\\
&=\ket{\mb{i}}\bra{\mb{i}}\otimes\bordermatrix{
&00&01&10&11\cr
00&1 & 1 & 1 &1\cr
01&1& 0& 1 & 0\cr
10&1 & 1 & 0 &0\cr
11&1& 0 & 0 & 1\cr
}.\\
&:=\ket{\mb{i}}\bra{\mb{i}}\otimes M_i^s
\end{aligned}
\end{equation}
for $\mb{i}=01,10$, and we denote the matrix between $j,k$ with a fixed $\mb{i}$ as $M_i^s$.

For simplicity of discussion, in the following we first consider a sub-ensemble with random CCZ gates defined as follows. For every three-qubit $i\in A, j,k \in \bar{A}$, we operate a CCZ on them or not with probability $1/2$, and the whole ensemble is composed of these local ones, similar as in Def.~\ref{Def:ensemble}. This ensemble is different from the original random CCZ ensemble, where $CCZ_{ijk}$ with $i,j\in A, k \in \bar{A}$ is also allowed, and thus we call it \emph{half} CCZ ensemble. See Fig.~\ref{Fig:halfCCZ} for an illustration. 

\begin{figure}[hbt]
\centering
\resizebox{7cm}{!}{\includegraphics[scale=0.8]{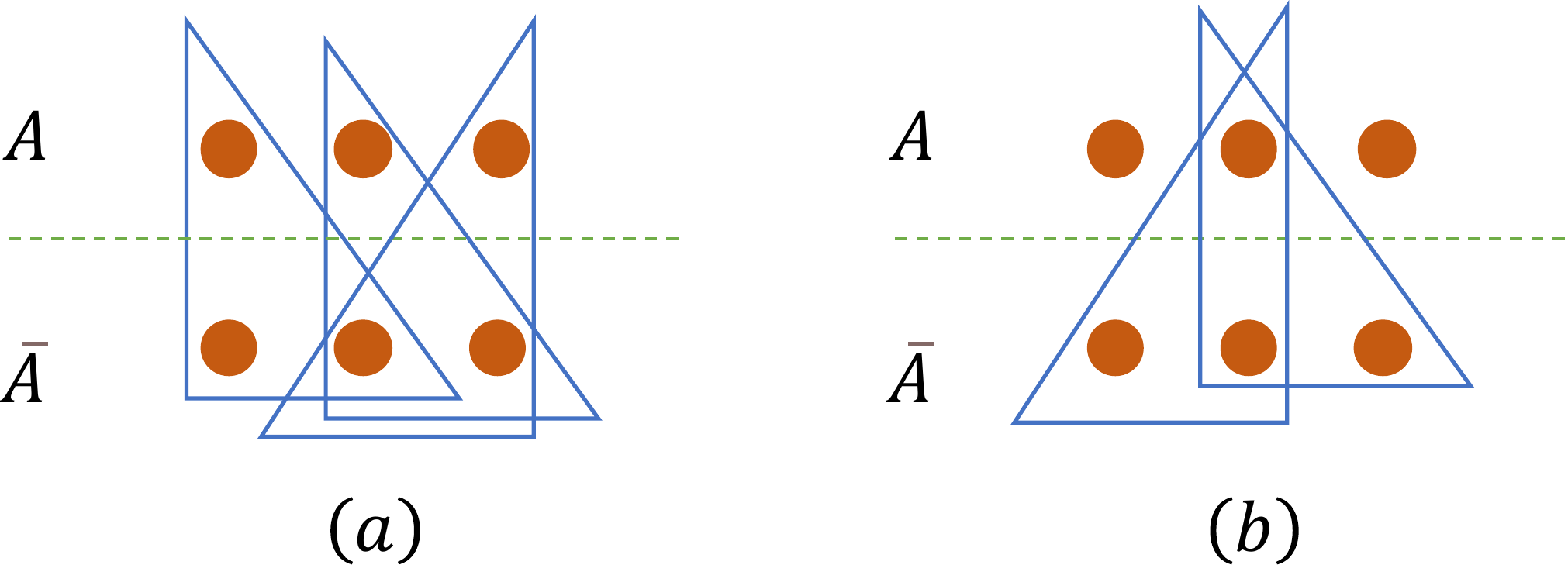}}
\caption{An illustration for the \emph{half} CCZ ensemble. In an $N=6$-qubit system, the subsystem $A$ and $\bar{A}$ both contain $N_A=N_{\bar{A}}=3$ qubits. For every three-qubit $i\in A, j,k \in \bar{A}$, we could operate a CCZ on them or not with $1/2$ probability. In this way, there are $N_A*C_{N_{\bar{A}}}^2$ possible gate patterns with equal probability. In (a) and (b) we show two examples for the gate pattern.}\label{Fig:halfCCZ}
\end{figure}

Similar as the CZ scenario, we will count the number of survival bit-strings of this half ensemble to calculate the average purity, with constraints from Eq.~\eqref{eq:CCZtrivial} and  \eqref{eq:MsPurCCZ}. By Eq.~\eqref{eq:CCZtrivial}, it is not hard to see that as $\mb{i}=\bar{0}$, actually there is no constraint on $j,k$; by Eq.~\eqref{eq:MsPurCCZ}, as $\mb{i}=\bar{1}$, some values of $j,k$ would vanish, shown as the zero elements in the matrix $M_i^s$.

First, if $\forall i\in A, \mb{i}=\bar{0}$, $j\in \bar{A}$ can take any value, totally $1*4^{N_{\bar{A}}}*2^{N_A}$, where  $2^{N_A}$ is the logical encoding abundance. Second, if there exists $\mb{i}=\bar{1}\in A$, there are three possibilities for $j\in \bar{A}$.

\begin{enumerate}[label=\textbf{\arabic*.},ref=\arabic*]
  \item $j$ only takes $00$ or $11$, totally $2^{N_{\bar{A}}}$. \label{CCZc1}
  \item there exists one $\mb{j}=01$ or $10$, and the remaining $k \in \bar{A}$ can just be $00$, totally $2C_{N_{\bar{A}}}^{1}=2{N_{\bar{A}}}$. \label{CCZc2}
  \item there exists one $\mb{j}=01$ and one $\mb{k}=10$, and the remaining qubits can just be $00$, totally $2C_{N_{\bar{A}}}^{2}=N_{\bar{A}}({N_{\bar{A}}}-1)$. \label{CCZc3}
\end{enumerate}

\comments{
\begin{enumerate}[label=\textbf{S.\arabic*},ref=S.\arabic*]
\item a
\item \label{l} b
\item c. goto \ref{l}
\end{enumerate}
}

As a result, the average purity of this half ensemble is obtained by summing the number of survival bit-strings in all the above cases and multiplying the normalization in Eq.~\eqref{eq:NormalPur}.
\begin{equation}\label{eq:avrPurHCCZ}
\begin{aligned}
\langle P_A \rangle_{CCZ,h}=2^{N_A}\frac{4^{N_{\bar{A}}}+(2^{N_A}-1)[2^{N_{\bar{A}}}+N_{\bar{A}}({N_{\bar{A}}}+1)]}{4^N}\\
=\frac{d_A+d_{\bar{A}}-1}{d}+\frac{d_A(d_A-1)N_{\bar{A}}({N_{\bar{A}}}+1)}{d^2}
\end{aligned}
\end{equation}
where the front $2^{N_A}$ is for the logical encoding abundance of the qubit in $A$.
In the case of equal partition $N_A=N_{\bar{A}}=N/2$, it shows $\langle P_A \rangle_{CCZ,h}\sim 2/\sqrt{d}+O(N^2/d)$.

\comments{
\begin{equation}\label{eq:PACCZhnn}
\begin{aligned}
\langle P_A \rangle_{CCZ,h}=(2\sqrt{d}-1)/d+\sqrt{d}(\sqrt{d}-1)n(n+1)/d^2\sim 2/\sqrt{d}+O(n^2/d),
\end{aligned}
\end{equation}
also similar with the Haar random case.
}

Finally, we look into the original full CCZ ensemble. Like before, we should additionally figure out the phase tensor for the case $i,j\in A$, and $k\in \bar{A}$. By the symmetry, the new tensor is almost the same to the original one. One has the relation of the phase tensor $\tilde{M}_{ijk}^p=M_{kij}^p$ by deforming the tensor in Fig.~\ref{Fig:PurCCZ} (c) to the one in (d), so does the summation tensor. As a result, the additional gate will induce more constraints on the bit-string, which leads to the \emph{decrease} of the average purity.

Actually, only the last two cases \ref{CCZc2} and \ref{CCZc3} listed before change. When there is $k\in \bar{A}$ taking $01$ or $10$, that is $\mb{k}=\bar{1}$, we can use the summation tensor $\tilde{M}_{ijk}^s=M_{kij}^s$ in Eq.~\eqref{eq:MsPurCCZ} to constraint the bit value of qubits back in $A$.
Since the case of $\forall i\in A, \mb{i}=\bar{0}$ has already be counted, one needs to figure out the case that  $\exists i \in A, \mb{i}=\bar{1}$. Following the same counting procedure in cases \ref{CCZc2} and \ref{CCZc3}, one can find that there are $N_A(N_A+1)$ choices for $A$. By summing the numbers of all these possibilities, one has the result of average purity as follows.
\comments{
As a result,
\begin{equation}
\begin{aligned}
\langle P_A \rangle_{CCZ}=\frac{[4^{N_{\bar{A}}}+(2^{N_A}-1)2^{N_{\bar{A}}}]2^{N_A}  +N_A(N_A+1)N_{\bar{A}}({N_{\bar{A}}}+1)}{4^N}=\frac{d_A+d_{\bar{A}}-1}{d}+\frac{N_A(N_A+1)N_{\bar{A}}({N_{\bar{A}}}+1)}{d^2}
\end{aligned}
\end{equation}
In the case of $N_A=N_{\bar{A}}=n$, it shows
\begin{equation}\label{eq:PACCZnn}
\begin{aligned}
\langle P_A \rangle_{CCZ}=(2\sqrt{d}-1)/d+n^2(n+1)^2/d^2\sim 2/\sqrt{d}-1/d+O(n^4/d^2),
\end{aligned}
\end{equation}
}
\begin{theorem}\label{th:CCZPur}
The average purity of the subsystem $A$ with $N_A$ qubits for the random hypergraph states from the CCZ ensemble is
\begin{equation}\label{eq:avrPurCCZ}
\begin{aligned}
\langle P_A \rangle_{CCZ}
=\frac{d_A+d_{\bar{A}}-1}{d}+\frac{N_A(N_A+1)N_{\bar{A}}({N_{\bar{A}}}+1)}{d^2}
\end{aligned}
\end{equation}
where $d_{A(\bar{A})}=2^{N_{A(\bar{A})}}$ is the Hilbert space dimension of the subsystem $A(\bar{A})$. For the case of equal partition $d_A=d_{\bar{A}}=\sqrt{d}$, one has $\langle P_A \rangle_{CCZ}=2/\sqrt{d}-1/d+O(N^4/d^2)$.
\end{theorem}

At the end of this section, we remark that the average purity results of CZ and CCZ ensembles obtained here are almost equal to that of Haar random states \cite{zyczkowski2001induced,Graeme2006Typical}, i.e.,
\begin{equation}\label{eq:PurHaar}
\begin{aligned}
\langle P_A \rangle_{\mathrm{Haar}}=\frac{d_A+d_{\bar{A}}}{d+1},
\end{aligned}
\end{equation}
and any state ensemble satisfies projective 2-design, such as the orbit of Clifford group \cite{divincenzo2002quantum,zhu2016clifford}. Actually, the CZ and CCZ ensemble are \emph{not} projective 2-design, even approximately \cite{Nakata2014,Nakata2014review}. 


\section{Fluctuations of the purity}\label{sec:variance}
In Sec.~\ref{sec:avrPur}, one sees that the average subsystem purities of random CZ and CCZ ensembles share the same leading term. In this section, we find that the variances of the purity of the two ensembles are quite different.
\begin{equation*}
\begin{aligned}
\delta^2(P_{A})=\mbb{E}_{\Psi}(P_{A}^2)-[\mbb{E}_{\Psi}(P_{A})]^2,
\end{aligned}
\end{equation*}
The essential quantity one needs to figure out is the first term $\mbb{E}_{\Psi}(P_{A}^2)$.
Similar as Eq.~\eqref{Eq:Pur} for the average purity, it can be written on the 4-copy Hilbert space $\mc{H}^{\otimes 4}$ as
\begin{equation}\label{eq:flucAll}
\begin{aligned}
\mbb{E}_{\Psi}(P_{A}^2)&=\mbb{E}_{\Psi}\left[\mathrm{Tr}\left(T_A\otimes \id_{\bar{A}}^{\otimes 2}\ \Psi^{\otimes 2}\right)\right]^2\\
&=\mbb{E}_{U\in\mc{E}}\ \ \mathrm{Tr}\left[T^{(1,2)}_A\otimes T^{(3,4)}_A \otimes \id_{\bar{A}}^{\otimes 4}  \ \ U^{\otimes 4}\Psi_0^{\otimes 4} U^{\dag\otimes 4}\right]\\
&=\mathrm{Tr}\left[T^{(1,2)}_A\otimes T^{(3,4)}_A\otimes \id_{\bar{A}}^{\otimes 4}  \ \ \Phi_\mc{E}^4\left(\Psi_0^{\otimes 2}\right)\right].
\end{aligned}
\end{equation}
Here we can still vectorize Eq.~\eqref{eq:flucAll} in the form of $\langle\langle O|\tilde{\Phi}|\rho\rangle\rangle$ as in Eq.~\eqref{Eq:PurVec}. $O=T^{(1,2)}_A\otimes T^{(3,4)}_A \otimes \id^{\otimes 4}_{\bar{A}}$, $\rho=\ket{+}\bra{+}^{\otimes 4N}$, and $\Phi$ is the 4-copy twirling channel $\Phi_\mc{E}^4$,
whose matrix form shows
\begin{equation}\label{Eq:4Tmatrix}
\begin{aligned}
\tilde{\Phi}_\mc{E}^4= (U\otimes U^*)^{\otimes 4}=U^{\otimes 8}
\end{aligned}
\end{equation}
with the fact that $U=U^*$ for phase gates. One can draw the corresponding diagram by \emph{doubling} Fig.~\ref{Fig:2qPur} (b).

Similar as the average purity in Fig.~\ref{Fig:2qPur} (c), we still apply the tensor network representation to calculate the expectation value here.
We arrange the total Hilbert space $\mc{H}^{\otimes 4}$ as $(\mc{H}_i^{\otimes 4} )^{\otimes N}$, i.e., put the two-copy of the i-th qubit Hilbert space $\mc{H}_i$ together in Fig.~\ref{Fig:FlucAll} (a). On the right, every qubit corresponds to 8 lines, each 2 for one-copy Hilbert space. In the middle, we operate the random $CZ_{e}$ gate across $A$ and $\bar{A}$. Here we show case of a CZ gate, which repeats 8 times by Eq.~\eqref{Eq:4Tmatrix}, connecting the corresponding lines of the two-qubit in the tensor network. On the left, by decomposing the swap $T^{(1,2)}_A\otimes T^{(3,4)}_A$ and the identity $\id^{\otimes 4}_{\bar{A}}$ to the qubit level, there are two different connections for each qubit determined by it belonging to $A$ or $\bar{A}$.

To calculate $\mbb{E}_{\Psi}(P_{A}^2)$, one needs to contract the tensor network in Fig.~\ref{Fig:FlucAll} (a) under the averge effect of the random $U$. Note that the input vector $(\ket{+}^{\otimes 8} )^{\otimes N}$ on the right of Fig.~\ref{Fig:FlucAll} (a) can take all the possible $0/1$ bit values with a normalization constant
\begin{equation}\label{eq:NormalVar}
\begin{aligned}
\mathcal{N}_v=(\frac1{\sqrt{2}})^{8N}=d^{-4}.
\end{aligned}
\end{equation}
Similar as the average purity scenario, one needs to count the number of bit-strings which survive under the twirling channel $\Phi_\mc{E}^4$, which can also be decomposed locally by Prop.~\ref{prop:ChannelLocal}. In fact, like Fig.~\ref{Fig:PurCZ} (a), one can associate 4 classical bit to one qubit due to the connections on the left in Fig.~\ref{Fig:FlucAll} (a), and we still apply the phase/summation matrix (tensor) to describe the local twirling channel $\Phi_{\mc{E}_e}^4$ on a edge $e$.

\begin{figure}[hbt]
\centering
\resizebox{9cm}{!}{\includegraphics[scale=0.8]{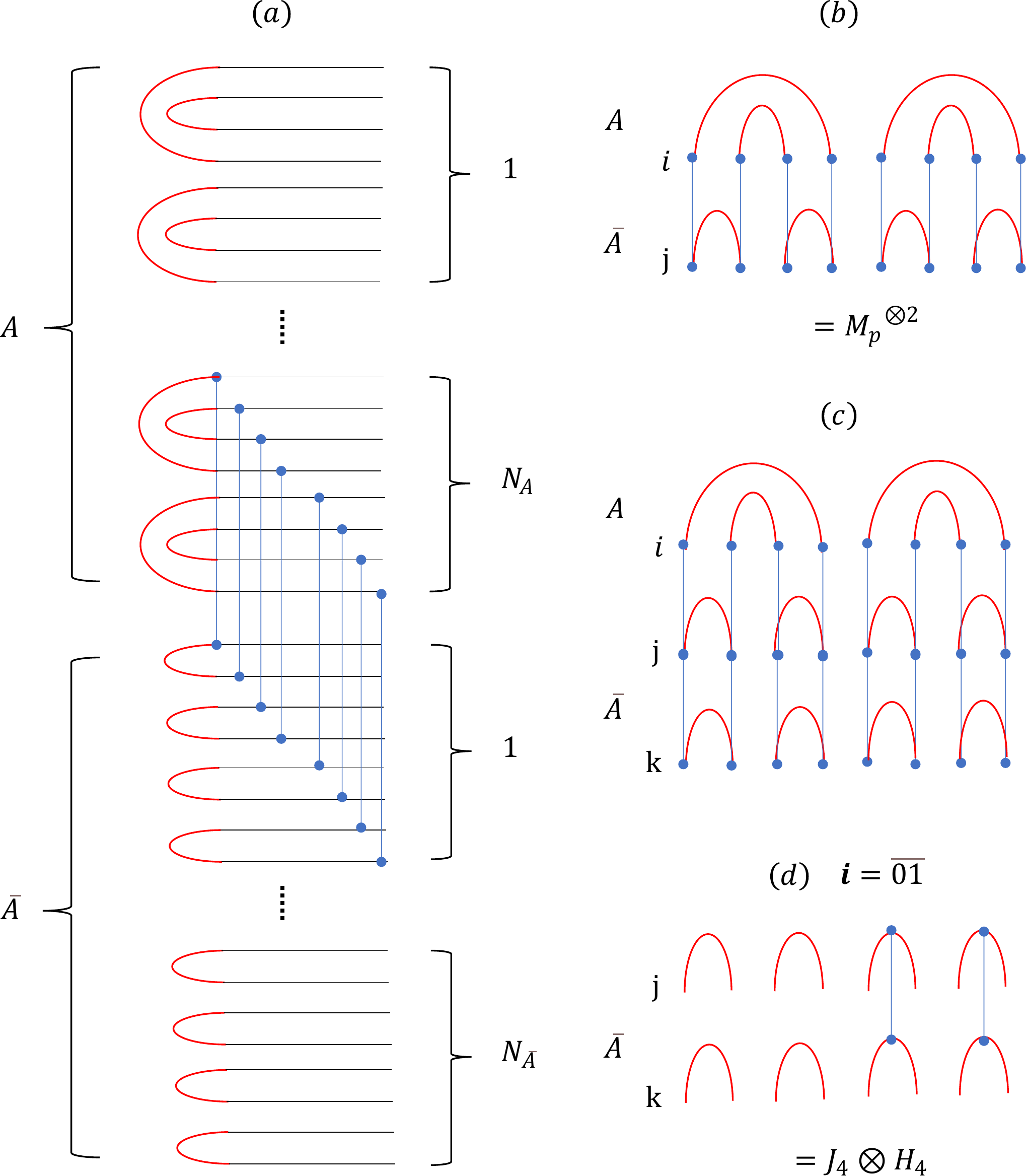}}
\caption{Tensor network for the purity variance formula and the phase tensors.  (a) We rearrange Eq.~\eqref{eq:flucAll} in the current tensor network form. 
On the right, every qubit corresponds to 8 lines, each 2 for one-copy Hilbert space. Every line can take one bit $0/1$ value.
In the middle, we can operate CZ or CCZ across $A$ and $\bar{A}$ in this double Hibert space. Here we show explicitly the CZ between the last qubit in $A$ and the first qubit in $\bar{A}$. On the left, there are two kinds of boundary operators correspond to swap and identity on the qubit Hilbert space $\mc{H}_i^{\otimes 4}$,  depending on whether the qubit-i is in $A$ or $\bar{A}$. (b) The phase matrix for the CZ gate, which is the 2-copy tensor product of the one in Fig.~\ref{Fig:PurCZ} (b) for the average purity scenario. (c) The phase tensor for the CCZ gate with qubit-i in A and j,k in $\bar{A}$, which is the tensor product of the one in Fig.~\ref{Fig:PurCCZ} (a). (d) The effective matrix on $j,k$ depends on the value of i. For example, if i takes $\bar{01}$, the matrix on j,k is $J_4\otimes H_4$ shown in Eq.~\eqref{Eq:PhMFluCCZ}.}\label{Fig:FlucAll}
\end{figure}


\subsection{Variance of purity for CZ ensemble}\label{sec:varCZ}
Similar as the average purity scenario in Sec.~\ref{CZpurity}, we should first figure out what happens locally, say the summation matrix for the twirling channel $\Phi_{\mc{E}_e}^4$.  Considering $e=\{i,j\}$ with $i \in A, j\in \bar{A}$, each qubit is labeled by 4-bit shown in Fig.~\ref{Fig:FlucAll} (b), where one can directly see that the phase matrix is just the tensor product of the one in Fig.~\ref{Fig:PurCZ} (b). By applying the logical encoding as in Eq.~\eqref{Eq:phMCZ}, each qubit can be labeled by 2-bit, and the phase matrix here shows
\begin{equation}\label{eq:MpVarCZ}
M_p'=M_p^{\otimes 2}=
\bordermatrix{%
&\bar{00}&\bar{01}&\bar{10}&\bar{11}\cr
\bar{00}&1 & 1 & 1 &1\cr
\bar{01}&1& -1& 1 & -1\cr
\bar{10}&1 & 1 & -1 &-1\cr
\bar{11}&1& -1 & -1 & 1\cr}.
\end{equation}
The summation matrix is obtained by average $M_p'$ with the all-ones matrix
\begin{equation}\label{}
M_s'=
\bordermatrix{%
&\bar{00}&\bar{01}&\bar{10}&\bar{11}\cr
\bar{00}&1 & 1 & 1 &1\cr
\bar{01}&1& 0& 1 & 0\cr
\bar{10}&1 & 1 & 0 &0\cr
\bar{11}&1& 0 & 0 & 1\cr}
\end{equation}

$M_s'$ shows the constraint on the classical bit configurations of the qubits in $A$ and $\bar{A}$. Like Sec.~\ref{CZpurity}, we list all the surviving possibilities on the logical encoding level as follows.
\begin{enumerate}[label=\textbf{\arabic*.},ref=\arabic*]
  \item $A$ only contains $\bar{00}$, and there is no constraint on the bit value of qubits in $\bar{A}$, and vice versa, totally $4^{N_A}+4^{N_{\bar{A}}}-1$.\label{c1:CZvar}
  \item only $\bar{00}$ and $\bar{11}$ appear in both $A$ and $\bar{A}$, totally $(2^{N_A}-1)(2^{N_{\bar{A}}}-1)$. \label{c2:CZvar}
  \item $A$ contains $\bar{01}$  and $\bar{00}$, and $\bar{A}$ contains $\bar{10}$ and $\bar{00}$, and vice versa, totally $2(2^{N_A}-1)(2^{N_{\bar{A}}}-1)$. \label{c3:CZvar}
\end{enumerate}
As a result, by summing the number of the possible bit-strings and normalizing the constant in Eq.~\eqref{eq:NormalVar}, one has
\begin{equation}\label{eq:P2CZ}
\begin{aligned}
\langle P_A^2 \rangle_{CZ}&=\frac{4^{N}[4^{N_A}+4^{N_{\bar{A}}}-1+3 (2^{N_A}-1)(2^{N_{\bar{A}}}-1)]}{4^{2N}}\\&=\frac{d_A^2+d_{\bar{A}}^2-1+3(d_A-1)(d_{\bar{A}}-1)}{d^2}.
\end{aligned}
\end{equation}
Here $4^{N}$ multiplication is due to the logical encoding redundancy. For example, $\bar{00}$ encoding for a qubit can correspond to 4 possibilities. As a result, by combing the average purity result in Eq.~\eqref{eq:avrPurCZ}, the variance shows as follows.
\begin{theorem}\label{th:CZVar}
For a subsystem $A$ with $N_A$ qubits, the variance of the purity $P_A$ defined in Eq.~\eqref{Eq:FlucDef} for the random graph states from the CZ ensemble is
\begin{equation}\label{eq:VarCZ}
\begin{aligned}
\delta^2_{CZ}(P_A)
&=\frac{(d_A-1)(d_{\bar{A}}-1)}{d^2}\\
\end{aligned}
\end{equation}
where $d_{A(\bar{A})}=2^{N_{A(\bar{A})}}$ is the Hilbert space dimension of the subsystem $A(\bar{A})$. For the case of equal partition $d_A=d_{\bar{A}}=\sqrt{d}$, one has $\delta^2_{CZ}(P_A)=d^{-1}-2d^{-\frac{3}{2}}+d^{-2}$.
\end{theorem}
We remark that the variance of the subsystem purity of CZ ensemble is similar to that of random stabilizer states \cite{Leone2021quantumchaosis}, which scales as $d^{-1}$ for the equal partition.

To get the above varaince, one subtracts $[\mbb{E}_{\Psi}(P_{A})]^2$ from $\mbb{E}_{\Psi}(P_{A}^2)$ in Eq.~\eqref{eq:P2CZ}.  Actually one is only left with the bit-strings in case \ref{c2:CZvar}, and the bit-strings in case \ref{c1:CZvar} and \ref{c3:CZvar} also survive in the formula of $[\mbb{E}_{\Psi}(P_{A})]^2$. The reason is as follows. For qubit $i$, it can be denoted by two logical bits $\mb{i_1}\mb{i_2}$, and the constraint from $M_p'$ in Eq.~\eqref{eq:MpVarCZ} on any qubit-pair i and j reads
\begin{equation}\label{eq:CZvar8times}
\begin{aligned}
\bra{\mb{i_1}}M_p\ket{\mb{j_1}}*\bra{\mb{i_2}}M_p\ket{\mb{j_2}}=1.
\end{aligned}
\end{equation}
On the other hand, the constraint of $[\mbb{E}_{\Psi}(P_{A})]^2$ for the two logical bits are independent, say $\bra{\mb{i_1}}M_p\ket{\mb{j_1}}=1$ \emph{and} $\bra{\mb{i_2}}M_p\ket{\mb{j_2}}=1$. As a result, the difference is the case where there exist two qubit $i,j$ such that the two terms  both take $-1$, which is just the case \ref{c2:CZvar}. This phenomena is general and also applicable to the CCZ scenario, and we illustrate it in detail in App.~\ref{ap:differ}.
\subsection{Variance of purity for CCZ ensemble}\label{CCZfluc}
Similar as the average purity scenario in Sec.~\ref{CCZpurity}, we first figure out the local phase and summation tensors. Considering three qubits $i\in A$ and $j,k\in\bar{A}$, as shown in Fig.~\ref{Fig:FlucAll} (c), each qubit is labeled by 4-bit.
It is clear that the phase tensor is the tensor-product of the one in the average purity scenario in Fig.~\ref{Fig:PurCCZ} (a),
\begin{equation}\label{eq:phMCCZvar}
\begin{aligned}
{M'}_{ijk}^{p}=M_{i_1j_1k_1}^p\otimes M_{i_2j_2k_2}^p
\end{aligned}
\end{equation}
with $i=\{i_1i_2\}$ and $i_1,i_2$ both taking 2-bit values, same for $k$ and $j$.

By using the logical encoding as in Sec.~\ref{CCZpurity}, $i_1$ and $i_2$ of qubit-i can be both labeled by 1 logical bit of $\bar{0},\bar{1}$. Depending on the status of $i$, the matrix between $j$ and $k$ shows,
\begin{equation}\label{Eq:PhMFluCCZ}
\begin{aligned}
&\mb{i}=\bar{00},\ J_4\otimes J_4=J_{16},\\
&\mb{i}=\bar{01},\ J_4\otimes H_4,\\
&\mb{i}=\bar{10},\ H_4\otimes J_4,\\
&\mb{i}=\bar{11},\ H_4\otimes H_4=H_{16},
\end{aligned}
\end{equation}
where $J_4$ denotes the trivial all-ones $4\times4$ matrix, and $H_4=H_2^{\otimes 2}$. See Fig.~\ref{Fig:FlucAll} (d) for the illustration of the second case $\mb{i}=\bar{01}$.

The summation tensor is obtained by averaging ${M'}_{ijk}^{p}$ with the trivial all-ones tensor $J_{ijk}$.
For a fixed $\mb{i}$, say $\mb{i}=\bar{01}$, the matrix between $j,k$ shows
\begin{equation}\label{}
\begin{aligned}
{M'}_i^s=\frac1{2}(J_4\otimes H_4+J_{16})
=J_4\otimes \frac1{2}(H_4+J_4)=J_4\otimes M^s_i
\end{aligned}
\end{equation}
where the explicit form of $M^s_i$ is given in Eq.~\eqref{eq:MsPurCCZ}. 
Similarly, $M_i'^s$ for different $\mb{i}$ is obtained by averaging every matrix in Eq.~\eqref{Eq:PhMFluCCZ}  with $J_{16}$, and they show respectively
\comments{
\begin{equation}\label{Eq:SumMFluCCZ}
{M'}_i^s=\left\{\begin{aligned}
&J_{16} &\mb{i}&=\bar{00},\\
&J_4\otimes M^s_i &\mb{i}&=\bar{01},\\
&M^s_i\otimes J_4 &\mb{i}&=\bar{10},\\
&\frac{1}{2}(H_{16}+J_{16}) &\mb{i}&=\bar{11}.
\end{aligned}\right.
\end{equation}
}

\begin{align}
&\mb{i}=\bar{00},\ {M'}_i^s=J_{16},\label{Eq:SumMFluCCZ1}\\
&\mb{i}=\bar{01},\ {M'}_i^s=J_4\otimes M^s_i,\label{Eq:SumMFluCCZ2}\\
&\mb{i}=\bar{10},\ {M'}_i^s=M^s_i\otimes J_4,\label{Eq:SumMFluCCZ3}\\
&\mb{i}=\bar{11},\ {M'}_i^s=\frac{1}{2}(H_{16}+J_{16}).\label{Eq:SumMFluCCZ4}
\end{align}

\comments{
\begin{subequations}  \label{eq:1}
\begin{align}  f &= g           \label{eq:1A} \\
              f' &= g'          \label{eq:1B} \\
   \mathcal{L}f  &= \mathcal{L}g \label{eq:1C}
\end{align}
\end{subequations}
}

In the following, as in the average purity scenario in Sec.~\ref{CCZpurity}, we first count the survival bit-string for the half ensemble of CCZ gates, say any $CCZ_{\{ijk\}}$ gate with $i\in A, j,k \in \bar{A}$, with the help of the summation tensor in Eq.~\eqref{Eq:SumMFluCCZ1} to \eqref{Eq:SumMFluCCZ4}, and extend to the full ensemble later. There are several possible cases listed as follows.
\begin{enumerate}[label=\textbf{\arabic*.},ref=\arabic*]
  \item $A$ only contains $\bar{00}$, $\bar{A}$ can be arbitrary by Eq.~\eqref{Eq:SumMFluCCZ1}, totally $16^{N_{\bar{A}}}$.\label{halfCCZVarC1}
  \item $A$ only contains $\bar{00}$ and $\bar{01}$. Due to Eq.~\eqref{Eq:SumMFluCCZ2}, the first two-bit of qubit in $\bar{A}$ can be arbitrary, and the last two should be restricted by $M^s_i$ in Eq.~\eqref{eq:MsPurCCZ}. Totally
\begin{equation}
\begin{aligned}
(2^{N_A}-1)4^{N_{\bar{A}}}\left[2^{N_{\bar{A}}}+N_{\bar{A}}(N_{\bar{A}}+1)\right]. \nonumber
\end{aligned}
\end{equation}
Here $(2^{N_A}-1)$ counts the possibility for $A$. 
$4^{N_{\bar{A}}}$ counts the possibility for the first two-bit of the qubit in $\bar{A}$.
$[2^{N_{\bar{A}}}+N_{\bar{A}}(N_{\bar{A}}+1)]$ accounts for the possibility for the last-two bit of qubits in $\bar{A}$, which follows similarly by summing cases \ref{CCZc1} to \ref{CCZc3} in Sec.~\ref{CCZpurity}.
\\Same result holds for $A$ only containing $\bar{00}$ and $\bar{10}$. 
\label{halfCCZVarC2}
\item $A$ 
contains at least two types from $\{\bar{01},\bar{10},\bar{11}\}$. For the qubit in $\bar{A}$, both the first two and last two-bit are restricted. For example, suppose $A$ has $\bar{01}$ and $\bar{11}$, the constraint follows Eq.~\eqref{Eq:SumMFluCCZ2} and \eqref{Eq:SumMFluCCZ4}. Consider two qubits $j,k\in \bar{A}$, the constraint on the classical 4-bit values is actually
 \begin{equation}\label{eq:halfCCZVarC3}
\begin{aligned}
\bra{\mb{j_1}}H_4\ket{\mb{k_1}}=1\\
\bra{\mb{j_2}}H_4\ket{\mb{k_2}}=1
\end{aligned}
\end{equation}
with $\mb{j_1},\mb{j_2}$ denoting the first and last-two bits of qubit j, similar for k. Note that the constraints are \emph{independent} on $\mb{j_1}$ and $\mb{j_2}$. Totally
\begin{equation}
\left[4^{N_A}-3(2^{N_A}-1)-1\right]\left[2^{N_{\bar{A}}}+N_{\bar{A}}(N_{\bar{A}}+1)\right]^2\nonumber
\end{equation}
where the first braket counts the possibility for $A$, and the second for $\bar{A}$ follows similar argument as in case \ref{halfCCZVarC2}. We count the possibility for the first and lat 2-bit independently and thus there is a square.
\label{halfCCZVarC3}
\item $A$ contains $\bar{00}$ and $\bar{11}$. For the qubit in $\bar{A}$, both the first two and last two-bit are restricted. 
Now the bit values of qubits $j,k\in \bar{A}$ is constrained \emph{jointly} by Eq.~\eqref{Eq:SumMFluCCZ4}. That is,
\begin{equation}\label{eq:CCZvarc4}
\begin{aligned}
\bra{\mb{j_1}}H_4\ket{\mb{k_1}}*\bra{\mb{j_2}}H_4\ket{\mb{k_2}}=1
\end{aligned}
\end{equation}
compared with Eq.~\eqref{eq:halfCCZVarC3} in case \ref{halfCCZVarC3}. We count this case explicitly as follows in Lemma \ref{Le:F24}.\label{halfCCZVarC4}
\end{enumerate}

Denote the 4-bit of the two qubit $j,k$ as $\{a_t\}$ and $\{b_t\}$ with $t=1,2,3,4$. The constraint in Eq.~\eqref{eq:CCZvarc4} is actually
\begin{equation}
\begin{aligned}
\prod_t\bra{a_t}H_2\ket{b_t}=\prod_t (-1)^{a_tb_t}=(-1)^{\sum_t a_tb_t}=1
\end{aligned}
\end{equation}
or equivalently the binary vector $\vec{a}$ and $\vec{b}$ are orthogonal on the binary filed. 
If we denote the number of bit configurations satisfy this condition as $\#_4$, one has the following result of it, with the proof in App.~\ref{proof:F24}.
\begin{lemma}\label{Le:F24}
Suppose there are $m$ distinct positions, one assigns 4-bit vector $\vec{v}_i\in \mbb{F}_2^4$ to each of them $1\leq i\leq m$, with the constraint
$$\vec{v}_i\cdot \vec{v}_j=0\ \mathrm{Mod}(2), \forall i\neq j.$$ The number of all possible assignments is 
\begin{equation}\label{Eq:L0011case4}
\begin{aligned}
\#_4=3*4^m+\Theta(m^2)2^m+\Theta(m^4)
\end{aligned}
\end{equation}
for large $m$.
\end{lemma}

As a result, by summing all the possibilities and normalizing the constant in Eq.~\eqref{eq:NormalVar}, one has
\onecolumngrid
\begin{equation}
\begin{aligned}
\langle P_A^2 \rangle_{CCZ,h}
&=4^{N_A}\frac{16^{N_{\bar{A}}}+2(2^{N_A}-1)4^{N_{\bar{A}}}\left[2^{N_{\bar{A}}}+N_{\bar{A}}(N_{\bar{A}}+1)\right]+\left[4^{N_A}-3(2^{N_A}-1)-1\right]\left[2^{N_{\bar{A}}}+N_{\bar{A}}(N_{\bar{A}}+1)\right]^2+(2^{N_A}-1)\#_4}{4^{2N}}\\
&=\frac{\left\{d(d_A+d_{\bar{A}}-1)+d_A(d_A-1)N_{\bar{A}}({N_{\bar{A}}}+1)\right\}^2+d_A^2(d_A-1)\left\{\#_4-\left[d_{\bar{A}}+N_{\bar{A}}(N_{\bar{A}}+1)\right]^2\right\}}{d^4}\\
\end{aligned}
\end{equation}
\twocolumngrid
with $4^{N_A}$ multiplication due to the logical encoding redundancy for qubit in $A$. As a result, by combing the average purity result in Eq.~\eqref{eq:avrPurHCCZ}, the variance of the half CCZ ensemble shows as follows.
\begin{theorem}\label{th:HCCZVar}
For a subsystem $A$ with $N_A$ qubits, the variance of the purity $P_A$ defined in Eq.~\eqref{Eq:FlucDef} for the random hypergraph states from the half CCZ ensemble is
\begin{equation}\label{eq:VarHCCZ}
\begin{aligned}
\delta^2_{CCZ,h}(P_A)&=\frac{d_A^2(d_A-1)\left\{\#_4-\left[d_{\bar{A}}+N_{\bar{A}}(N_{\bar{A}}+1)\right]^2\right\}}{d^4}\\
&\le \frac{9d_A}{d^2}\le 9d^{-\frac3{2}},
\end{aligned}
\end{equation}
where $\#_4$ is in Lemma \ref{Le:F24} by taking $m=N_{\bar{A}}$, and $d_{A(\bar{A})}=2^{N_{A(\bar{A})}}$ is the Hilbert space dimension of the subsystem $A(\bar{A})$.
\end{theorem}

Similar as the discussion at the end of Sec.~\ref{CCZfluc}, from Eq.~\eqref{eq:VarHCCZ}, it is not hard to see that only the case \ref{halfCCZVarC4} contributes to the final variance. Actually the bit-strings in cases $\ref{halfCCZVarC1}$ to $\ref{halfCCZVarC3}$ have already been counted for $[\mbb{E}_{\Psi}(P_{A})]^2$ of Eq.~ \eqref{eq:avrPurHCCZ} in Sec.~\ref{CCZpurity}. In particular, the $\left[d_{\bar{A}}+N_{\bar{A}}(N_{\bar{A}}+1)\right]^2$ in Eq.~\eqref{eq:VarHCCZ} indicates that one should further substrate the solution of Eq.~\eqref{eq:CCZvarc4} with $\bra{\mb{j_1}}H_4\ket{\mb{k_1}}=\bra{\mb{j_2}}H_4\ket{\mb{k_2}}=1$ for all $j,k$ pairs.

Before moving to the full CCZ ensemle, we would like to remark that as one extends the random CCZ gates from the half ensemble to the full ensemble, the average value should decrease
\begin{equation}
\begin{aligned}
\langle P_A^2 \rangle_{CCZ}\leq \langle P_A^2 \rangle_{CCZ,h},
\end{aligned}
\end{equation}
since the full ensemble will induce more constraints on the bit configuration and less of them can survive.
As a result, by inserting the average purity in Eq.~\eqref{eq:avrPurCCZ}, one has the result for the full ensemble, with the proof left in App.~\ref{ap:proofCoCCZ}.
\begin{corollary}\label{co:CCZVar}
For a subsystem $A$ with $N_A$ qubits, the variance of the purity $P_A$ defined in Eq.~\eqref{Eq:FlucDef} for the random hypergraph states from the CCZ ensemble is upper bounded by
\begin{equation}\label{eq:VarCCZupp}
\begin{aligned}
\delta^2_{CCZ}(P_A)< \langle P_A^2 \rangle_{CCZ,h}- \langle P_A \rangle_{CCZ}^2<3N^2 d^{-\frac3{2}}
\end{aligned}
\end{equation}
for sufficient large $N$ and $d=2^N$ is the total Hilbert space dimension.
\end{corollary}
Furthermore, one can enhance the bound by more delicate counting as in the average purity scenario in Sec.~\ref{CCZpurity}, and show that the variance of the full ensemble scales as $d^{-2}$, compared to $d^{-1.5}$ in Corollary \ref{co:CCZVar}.
\begin{theorem}\label{th:CCZVar}
For a subsystem $A$ with $N_A$ qubits, the variance of the purity $P_A$ defined in Eq.~\eqref{Eq:FlucDef} for the random hypergraph states from the CCZ ensemble is
\begin{equation}\label{eq:VarCCZ}
\begin{aligned}
\delta^2_{CCZ}(P_A)= 4d^{-2}-2(d_A+d_{\bar{A}})d^{-3}+O(N^4)d^{-3}
\end{aligned}
\end{equation}
where $d=2^N$ is the total Hilbert space dimension.
\end{theorem}
The proof is left in App.~\ref{ap:proofCCZ}. 
Note that the variance of the subsystem purity of CCZ ensemble is similar to that of Haar random states \cite{Hamma2012Ensembles,Leone2021quantumchaosis}, which scales as $d^{-2}$ for the equal partition.

At the end of this subsection, we give some remark on the generalization of our method beyond the uniform ensemble of the $CZ_{e}$ gate in Def.~\ref{Def:ensemble}. First, suppose the hypergraph state ensemble is not uniform, say the probability $p\neq 1/2$, one can still apply the phase/summation matrix formalism, and now $M_s$ say in Eq.~\eqref{eq:sumMfull}  become $M_s=pM_p+(1-p)J$, which could make the calculation more involved. But we conjecture that any constant deviation of $p$ from $1/2$ will not essentially change the previous results on the purity and its variance. Generally, the $M_s$ matrix could be a summation or integral like
\begin{equation}\label{eq:MsGeneral}
\begin{aligned}
M_s=\sum_i p_i M_p^{(i)}
\end{aligned}
\end{equation}
where $M_p^{(i)}$ is the phase matrix for some phase gate with applying probability $p_i$. In our original uniform ensemble, the gates are $\{\id_{e},CZ_{e}\}$, both with $1/2$ probability. Moreover, our formalism could also be suitable for diagonal gates, since they keep the transformation on the computational basis.

\section{Implications for the entanglement entropy}
In this section, we utilize the results of the average subsystem purity and its variance obtained in Sec.~\ref{sec:avrPur} and \ref{sec:variance}, to show some statistical behaviour of the entanglement entropy.

Recall that the Von Neumann entropy can be bounded by the  R\'enyi-2 entropy
\begin{equation}\label{eq:VNlower}
\begin{aligned}
\mbb{E}_{\Psi}S_1(\rho_A)&\geq \mbb{E}_{\Psi} S_2(\rho_A)\\
&=\mbb{E}_{\Psi}-\log_2(P_A)\\
&\geq -\log_2(\mbb{E}_{\Psi}P_A),
\end{aligned}
\end{equation}
where the last line is due to the convexity.
Consequently, by Theorem \ref{th:CZPur} and \ref{th:CCZPur}, one has the following result.
\begin{corollary}\label{co:CZCCZEnt}
For a subsystem $A$ with $N_A$ qubits, the average  R\'enyi-2 entanglement entropy of the random hypergraph states from both CZ and CCZ ensembles is lower bounded by
\begin{equation}\label{eq:CZCCZEnt}
\begin{aligned}
&\langle S_2(\rho_A) \rangle
\geq -\log_2\left(\frac{d_A+d_{\bar{A}}}{d}\right)\\
\end{aligned}
\end{equation}
where $d_{A(\bar{A})}=2^{N_{A(\bar{A})}}$ is the Hilbert space dimension of the subsystem $A(\bar{A})$, and so does the average Von Neumann entropy $\langle S_1(\rho_A) \rangle$ by Eq.~\eqref{eq:VNlower}. 
\end{corollary}
Some remarks are as follows. In the regime $N_{\bar{A}}-N_A\gg 1$, it is clear that the average entanglement entropy for both ensembles equals the qubit number in the subsystem $A$, $\langle S_2(\rho_A) \rangle\sim -\log(d_A^{-1})=N_A$, which is the largest possible value. In the regime $N_A=N_{\bar{A}}=N/2$, the lower bound of average entropy reads  $-\log(2d^{-1/2})=N/2-1$, which shows a constant departure to the maximal possible value. This subtlety is also observed for random stabilizer states \cite{dahlsten2005exact}, and Haar random states \cite{hayden2006aspects}.

Moreover, we can apply the variance of the purity function $P_A$ to show the typical behaviour. 
By Chebyshev inequality,
\begin{equation}\label{ChebyshevPA}
\begin{aligned}
    \mathrm{Pr}\{|P_A-\langle P_A \rangle|> \varepsilon\} \le \frac{\delta^2(P_A)}{\varepsilon^2}.
   \end{aligned}
\end{equation}
and one can further obtain the following result on entanglement entropy.

\begin{prop}\label{prop:deviation}
For a subsystem $A$ with $N_A$ qubits, the probability of the entanglement entropy deviation is upper bounded by
\begin{equation}\label{eq:prop:deviation}
\begin{aligned}
\mathrm{Pr}\left\{S_2(\rho_A)\leq -\log_2\left(\frac{d_A+d_{\bar{A}}}{d}\right)-1.5\varepsilon d_A\right\}\leq \frac{\delta^2(P_A)}{\varepsilon^2},
\end{aligned}
\end{equation}
where $d_{A(\bar{A})}=2^{N_{A(\bar{A})}}$ is the Hilbert space dimension of the subsystem $A(\bar{A})$, and $\delta^2(P_A)$ is the variance of the subsystem purity $P_A$ of CZ or CCZ ensembles.
\end{prop}

By Theorem \ref{th:CZVar} and Theorem \ref{th:CCZVar}, the variances $\delta^2_{CZ}(P_A)=\Theta(d^{-1})$ and $\delta^2_{CCZ}(P_A)=\Theta(d^{-2})$, no matter the relative size between $N_A$ and $N_{\bar{A}}$.
In the regime $d_A\ll d_{\bar{A}}$, 
one can choose $\delta(P_A) \ll \varepsilon \ll d_A^{-1}$ for both ensembles such that the amount of the deviation, i.e., $1.5\varepsilon d_A$, from the maximal entropy is negligible, and also the probability of this  deviation i.e., $\delta^2(P_A)/\varepsilon^2$,  is exponentially small with respective to the qubit number $N$. As a result, one can further show exponentially small variance of the entanglement entropy. We leave the proof of Prop.~\ref{prop:deviation} and the detailed discussion in App.~\ref{ap:entfluc1}. 

In the regime $d_A\simeq d_{\bar{A}}$ i.e., the subsystem size is comparable $N_{\bar{A}}-N_A=O(\log(N))$, $d_A \sim d^{-\frac{1}{2}}$ now.
The previous argument can still apply for the CCZ ensemble, for example, by choosing $d^{-1}\ll \varepsilon\sim d^{-\frac{3}{4}}\ll d^{-\frac1{2}}$ to make the deviation and the corresponding probability both exponentially small. We summarize the result for the equal partition for conciseness as follows.
\begin{theorem}\label{th:CCZEntvar}
Given a system of $N$-qubit with $N$ sufficient large, and the equal subsystem size $N_A=N_{\bar{A}}$, for the random hypergraph states from the CCZ ensemble,
the probability of the entanglement entropy deviation is upper bounded by
\begin{equation}\label{eq:CCZEntProb}
\begin{aligned}
\mathrm{Pr}\left\{\left|S_2(\rho_A)-\left(\frac{N}{2}-1\right)\right|\geq 1.6\varepsilon d^{\frac1{2}} \right\}\leq \frac{4d^{-2}}{\varepsilon^2},
\end{aligned}
\end{equation}
where $d=2^{N}$ is the total Hilbert space dimension and $d^{-1}\ll \varepsilon\ll d^{-\frac1{2}}$. 
Consequently, by taking $\varepsilon\sim d^{-\frac{3}{4}}$, the variance of the entanglement entropy with respective to the CCZ ensemble is bounded by 
\begin{equation}\label{eq:CCZEntvar}
\begin{aligned}
\mathrm{Var}_{CCZ}[S_2(\rho_A)]<1.6Nd^{-\frac1{2}}.
\end{aligned}
\end{equation}
\end{theorem}
The proof of Theorem \ref{th:CCZEntvar} is left in App.~\ref{ap:entfluc2}. It is not hard to see that such exponentially small variance of entanglement entropy also holds for the half CCZ ensemble, as the standard variance of purity is about $\delta_{CCZ,h}(P_A)\sim d^{-3/4}\ll d^{-1/2}$ in Theorem \ref{th:HCCZVar}. We remark that similar concentration of measure bound on entanglement entropy of Haar random states is also shown in Ref.~\cite{hayden2006aspects}.

However, it is not applicable to the CZ ensemble now, since $\delta_{CZ}(P_A)=\Theta(d^{-1/2})$ and $d_A\sim d^{1/2}$, i.e., the standard variance is comparable to the expectation value of purity $\langle P_A \rangle\sim d^{-1/2}$ by Theorem \ref{th:CZPur}. There is no room to choose a appropriate $\varepsilon$ such that the deviation and the corresponding probability in Eq.~\eqref{eq:prop:deviation} both keep exponentially small. Especially in the equal partition case,
one thus expects a \emph{constant} variance of the entanglement entropy, which is proved in the following Theorem.
\begin{theorem}\label{th:CZEntvar}
Given a system of $N$-qubit with the equal subsystem size $N_A=N_{\bar{A}}$, for the random graph states generated by random CZ gates, the variance of the entanglement entropy
\begin{equation}\label{eq:CZEntvar}
\begin{aligned}
\mathrm{Var}_{CZ}[S_2(\rho_A)]>0.128,
\end{aligned}
\end{equation}
for sufficient large $N$.
\end{theorem}
The proof is based on the statistical result of the rank distribution for the random binary matrix \cite{kolchin_1998}, and we leave it in App.~\ref{ap:entfluc3}.

\section{concluding remarks}
In this work, we study the entanglement properties of random hypergraph states generated by CZ and CCZ gates, respectively. We find that, though the average subsystem purity and entanglement  both feature the same volume law,  fluctuations of  entanglement in the CZ ensemble are large, while the CCZ ensemble shows typical values of entanglement with vanishing fluctuations. These results show that, in spite of CCZ gates not being universal, they feature universal entanglement behavior \cite{Leone2021quantumchaosis}.

In perspective, there are several directions to explore starting from the results presented here. First, one could study general hypergraph states with $k>3$ edges using the phase matrix method here. The extension to more general diagonal gates \cite{Kruszynska2009Local,Nakata2014review,iaconis2021quantum}, qudit cases \cite{steinhoff2017qudit,xiong2018qudit} and mixed states \cite{Graeme2006Typical,wu2014randomized} are also possible.  Second, in this work we focus on the bipartite entanglement, and it is interesting to study the tripartite and even complex entanglement structure \cite{Guhne2005Multipartite,zhou2019detecting} and entanglement distillation yield \cite{Graeme2006Typical} of the random hypergraph state, which would supply useful tools to quantum networks. Third, the main difference between CCZ and CZ gates is that CCZ gates are non-Clifford, and as such they can produce `magic'. The relationship between fluctuations of the purity and magic has been shown in \cite{Leone2021quantumchaosis,leone2021r}. On the other hand, non-Clifford gates are involved in the onset of entanglement complexity and emergent irreversibility in unitary evolution \cite{Chamon2014Irreversibility}. In this context, it would be interesting to show the transition to quantum chaos behavior by doping a CZ circuit with CCZ gates \cite{zhou2020single,haferkamp2020quantum,Leone2021quantumchaosis}. Moreover, it would thus be very important to show that the fluctuations of the purity are directly responsible for the onset of entanglement complexity and the impossibility of undoing entanglement by Metropolis-like algorithms \cite{Chamon2014Irreversibility,shaffer2014irreversibility}. For the same reason, it would be interesting to study the behavior of (higher order) out-of-time-order correlation functions \cite{Shenker_2014,Yoshida2016Topological} under CZ and CCZ circuits, with and without doping. Finally, it would be intriguing to see whether our techniques here can be useful to study many-qubit magic states \cite{liu2020manybody,leone2021r}.

\section{acknowledgements}
We thank Mile Gu, Xun Gao, Arthur Jaffe and Zhengwei Liu for the useful discussions. Y.Z. is supported by the start-up funding of Fudan Univ., the Quantum Engineering Program QEP-SF3, National Research
Foundation of Singapore under its NRF-ANR joint program (NRF2017-NRF-ANR004 VanQuTe), the Singapore Ministry of Education Tier 1 grant RG162/19, FQXi-RFP-IPW-1903 from the foundational Questions
Institute and Fetzer Franklin Fund, a donor advised fund of Silicon Valley Community Foundation. A.H. acknowledges support from NSF award number 2014000. Any opinions, findings and conclusions or recommendations expressed in this material are those of the author(s) and do not reflect the views of the National Research Foundation, Singapore.


%

\appendix

\onecolumngrid
\newpage
\section{The relation between $M_s$ and $\Phi_{\mc{E}_{e}}^k$}\label{ap:sumM&TrChannel}
In main text, we use the summation matrix (tensor) $M_s$ to account for the effect of the local twirling channel $\Phi_{\mc{E}_{e}}^k$, with $k=2$ and $4$ for the average purity and the variance respectively. Here we clarify the relation between $M_s$ and $\Phi_{\mc{E}_{e}}^k$ in detail by considering the example for $k=2$ and the random CZ gate. The argument applies to other cases.

Recall in Fig.~\ref{Fig:2qPur} (c) in main text, the formula of average purity has been vectorized. The input state is $\ket{+}^{\otimes 2N}$ which can exhaust all the computational basis input. The matrix representation of the twirling channel $\tilde{\Phi}_{\mc{E}}^2$ is diagonal in the computational basis. And it can be decomposed to local twirling channels $\tilde{\Phi}_{\mc{E}_{e}}^2$, with $e$ for a specific edge $e=\{i,j\}$, by Prop.~\ref{prop:ChannelLocal}. The tensor diagram is shown in Fig.~\ref{MsPhi}, which is the average of the two possibilities $CZ_e$ and $\id_e$. Every line can take take $0/1$, and thus
$\tilde{\Phi}_{\mc{E}_{e}}^2$ is $8\times8$ diagonal matrix.  To get $\tilde{\Phi}_{\mc{E}_{e}}^2$, one just needs to figure out $8$ diagonal elements there.

\begin{figure}[hbt]
\centering
\resizebox{10cm}{!}{\includegraphics[scale=0.8]{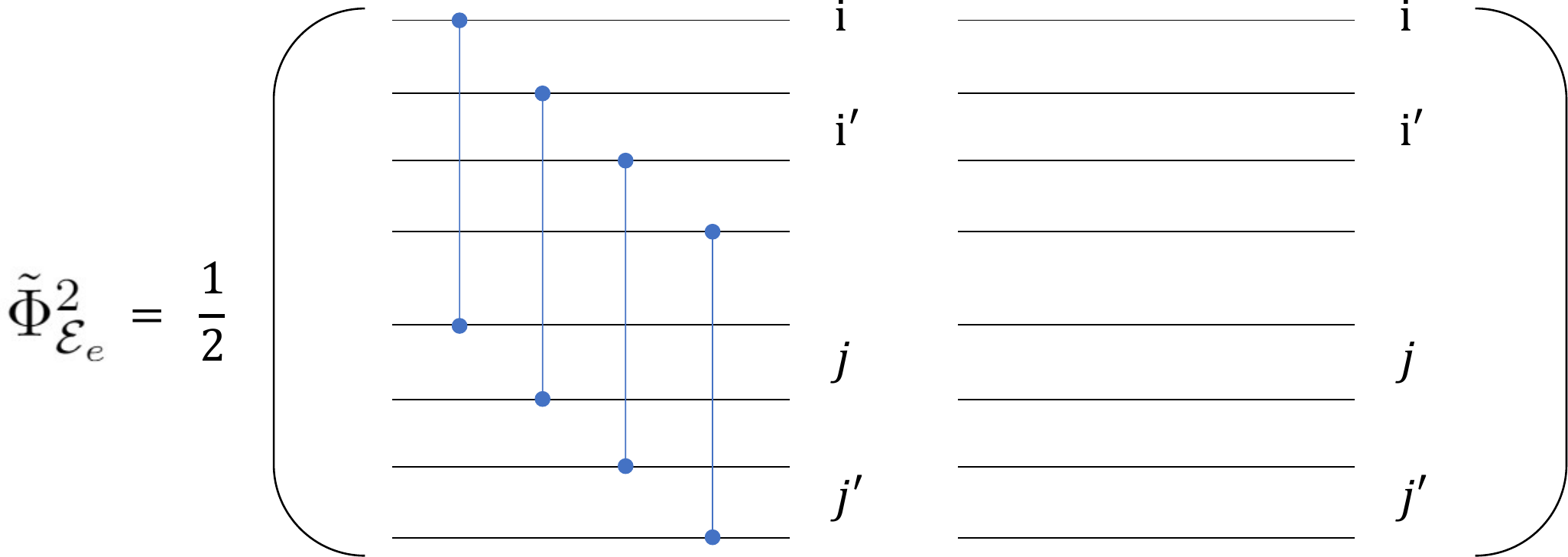}}
\caption{Tensor network for the matrix form of the local twirling channel $\tilde{\Phi}_{\mc{E}_{e}}^2$}\label{MsPhi}
\end{figure}
Actually, for the average purity we consider, one needs not to figure out all of them.
The vectorized swap and identity operators on the left in Fig.~\ref{Fig:2qPur} (c) further restrict the possible computational basis input. As illustrated in Fig.~\ref{Fig:PurCZ} (a), the 4 bits on the line of $i$ and $i'$ which correspond to qubit-$i$ will be reduced to 2 bits, similar to the qubit-$j$.

By further taking $i,i'$ as the row index, and $j,j'$ as the column index, we reach the phase matrix in Fig.~\ref{Fig:PurCZ} (b) and also the summation matrix $M_s$ in Eq.~\eqref{eq:sumMfull}.

\section{The difference between $\mbb{E}_{\Psi}(P_{A}^2)$ and $[\mbb{E}_{\Psi}(P_{A})]^2$ in the view of phase matrix}\label{ap:differ}
\begin{figure}[hbt]
\centering
\resizebox{10cm}{!}{\includegraphics[scale=0.8]{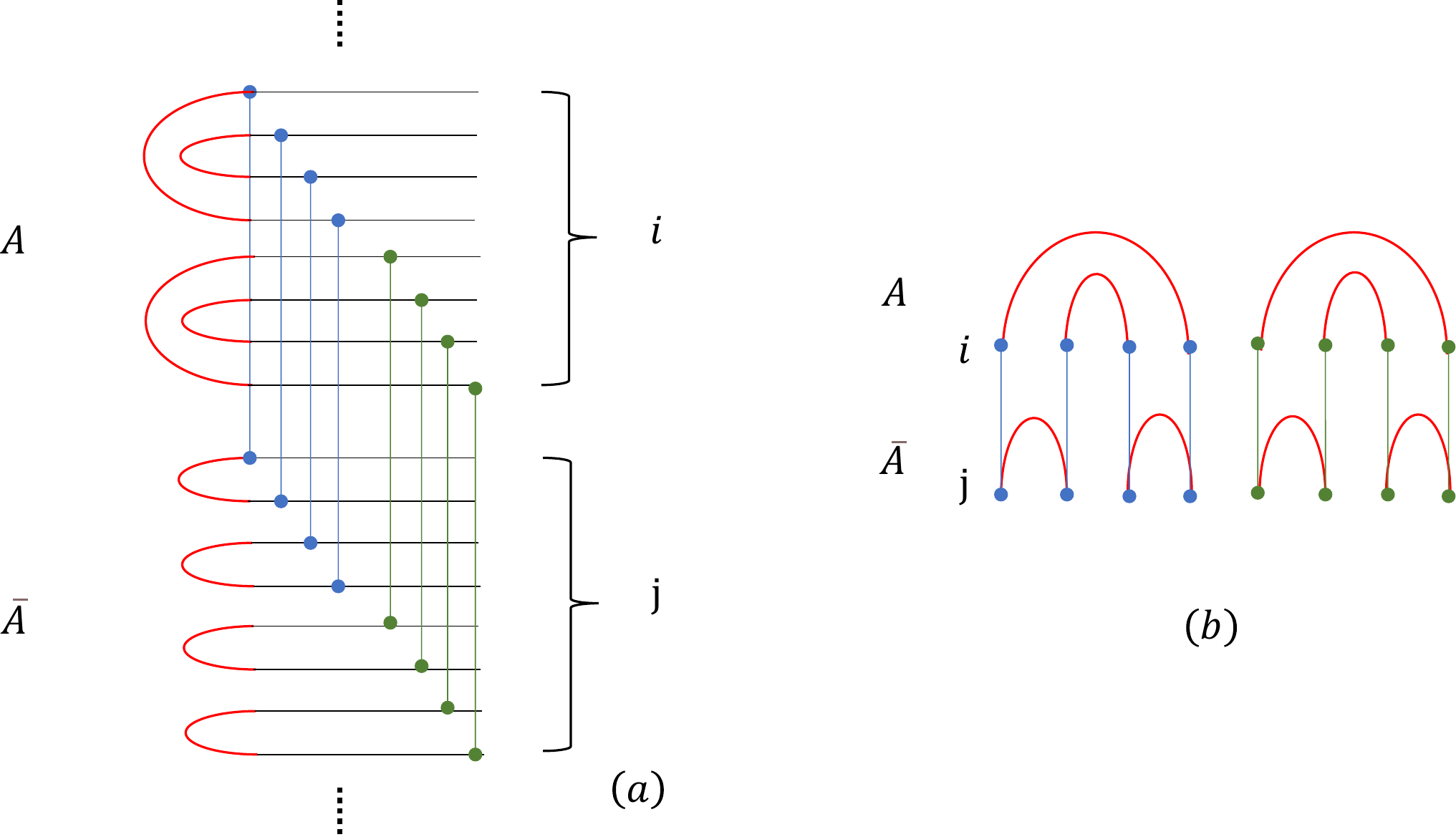}}
\caption{Tensor network for $[\mbb{E}_{\Psi}(P_{A})]^2$ and the phase matrix. }\label{Fig:appSquare}
\end{figure}
The difference of $\mbb{E}_{\Psi}(P_{A}^2)$ and $[\mbb{E}_{\Psi}(P_{A})]^2$ is the variance in Eq.~\eqref{Eq:FlucDef}. Here we illustrate the difference with phase matrix formalism in detail, which can simplify the counting procedure of the surviving bit-strings when calculating the variance.

By doubling Eq.~\eqref{Eq:Pur}, let us write down $[\mbb{E}_{\Psi}(P_{A})]^2$ explicitly  on 4-copy Hilbert space $\mc{H}^{\otimes 4}$ as
\begin{equation}\label{app:eqVarsq}
\begin{aligned}
\left[\mbb{E}_{\Psi}(P_{A})\right]^2&=\mathrm{Tr}\left[T_A\otimes \id_{\bar{A}}^{\otimes 2}  \ \Phi_\mc{E}^2\left(\Psi_0^{\otimes 2}\right)\right]^2\\
&=\mathrm{Tr}\left[T^{(1,2)}_A\otimes T^{(3,4)}_A\otimes \id_{\bar{A}}^{\otimes 4}\quad  \Phi_\mc{E}^2\otimes \Phi_\mc{E}^2 \left(\Psi_0^{\otimes 4}\right)\right].
\end{aligned}
\end{equation}
Compared with $\mbb{E}_{\Psi}(P_{A}^2)$ in Eq.~\eqref{eq:flucAll}, the input state $\Psi_0^{\otimes 4}$ and the swap/identity operators are the same, and the only difference is the twirling channel in the middle. In particular, one is $\Phi_\mc{E}^2\otimes \Phi_\mc{E}^2$, where the first 2-copy and the last 2-copy are independently twirled and the other is $\Phi_\mc{E}^4$, with 4-copy twirled together. Similar as the proof for Prop.~\ref{prop:ChannelLocal}, $\Phi_\mc{E}^2\otimes \Phi_\mc{E}^2$ can also be decomposed to the local twirling channel on a specific edge $\Phi_{\mc{E}_e}^2\otimes \Phi_{\mc{E}_e}^2$.

In the following, we use CZ ensemble to illustrate the difference of these two twirling channels, in term of the phase matrix. The result on CCZ ensemble follows similarly. Similar as Fig.~\ref{Fig:FlucAll} (a) of $\mbb{E}_{\Psi}(P_{A}^2)$, we can draw the tensor network diagram of $\left[\mbb{E}_{\Psi}(P_{A})\right]^2$ of Eq.~\eqref{app:eqVarsq} in Fig.~\ref{Fig:appSquare} (a). On the right, every qubit corresponds to 8 lines, each 2 for one-copy Hilbert space. Every line can take can take one bit $0/1$ value with a normalization constant in Eq.~\eqref{eq:NormalVar}. On the left, there are two different connections for each qubit by decomposing the swap $T^{(1,2)}_A\otimes T^{(3,4)}_A$ and the identity $\id^{\otimes 4}_{\bar{A}}$ to the qubit level. The difference is in the middle, where we operate the random $CZ_{e}$ gate across $A$ and $\bar{A}$, and we show case of the gates between $i$ and $j$. Note that two different colors (blue and green) are used to denote the independence of the CZ gates on the first 2-copy and the last 2-copy.

In Fig.~\ref{Fig:appSquare} (b), we rotate the diagram to compare it with the previous phase matrix in Fig.~\ref{Fig:FlucAll} (b), where the CZ gate repeats 8 times. Similar as the discussion Sec.~\ref{sec:varCZ}, every qubit-i is denoted by 4-bit and further by 2-bit $i_1i_2$ with logical encoding. Now the constraints between any qubit pair $i,j$ across $A,\bar{A}$ becomes
\begin{equation}
\begin{aligned}
&\bra{\mb{i_1}}\bra{\mb{i_2}}M_s\otimes M_s\ket{\mb{j_1}}\ket{\mb{j_2}}\\
=&\bra{\mb{i_1}}\bra{\mb{i_2}}\frac{M_p+J_2}{2}\otimes \frac{M_p+J_2}{2}\ket{\mb{j_1}}\ket{\mb{j_2}}\\
\end{aligned}
\end{equation}
with $M_s$ and $M_p$ defined in Eq.~\eqref{Eq:sumMCZ} and \eqref{Eq:phMCZ}. That is,
\begin{equation}\label{app:ConstCZ2}
\begin{aligned}
\bra{\mb{i_1}}M_p\ket{\mb{j_1}}=1,\ \bra{\mb{i_2}}M_p\ket{\mb{j_2}}=1.
\end{aligned}
\end{equation}
On the other hand, the phase matrix of Eq.~\eqref{eq:MpVarCZ} in Fig.~\ref{Fig:FlucAll} (b) induces constraint as
\begin{equation}
\begin{aligned}
\bra{\mb{i_1}}\bra{\mb{i_2}}\frac{M_p^{\otimes 2}+J_4}{2}\ket{\mb{j_1}}\ket{\mb{j_2}},
\end{aligned}
\end{equation}
i.e.,
\begin{equation}\label{app:ConstCZVar}
\begin{aligned}
\bra{\mb{i_1}}M_p\ket{\mb{j_1}}*\bra{\mb{i_2}}M_p\ket{\mb{j_2}}=1
\end{aligned}
\end{equation}
Compared with Eq.~\eqref{app:ConstCZ2}, this constraint is loose. That is, any bit-string configuration satisfies Eq.~\eqref{app:ConstCZ2} also survives here, which leads to $\delta^2=\mbb{E}_{\Psi}(P_{A}^2)-\left[\mbb{E}_{\Psi}(P_{A})\right]^2\geq 0$, a non-negative variance. As a result, to calculate $\delta^2$, one only needs to count the net number of bit-strings which satisfy
Eq.~\eqref{app:ConstCZVar}, but not Eq.~\eqref{app:ConstCZ2}. That is, for a bit-string configuration, there exists qubit-pair $i,j$, such that $\bra{\mb{i_1}}M_p\ket{\mb{j_1}}=-1$ and $\bra{\mb{i_2}}M_p\ket{\mb{j_2}}=-1$.

\comments{

\red{discuss the difference later} which is different on the subleading order compared to the half CCZ case.
According to the second line of Eq.~\eqref{Eq:Pur}, if an ensemble satisfies the projective 2-design, it can reproduce the average purity as Haar measure. However, our ensemble actually is not a 2-design, even approximately. Even though there is some approximate result by Nakata, it is not sufficient to study the average and its fluctuation here. \red{make a comparation in App}

Let's do it. happy lalala
}

\section{Proof of Lemma \ref{Le:F24}}\label{proof:F24}
\begin{proof}

To satisfy the inner product on the binary field, generally one can not assign the same binary vector $\vec{v}$ with an odd hamming weight (the number of $1$s), denoted by $w_2(\vec{v})=1$, to more than one position, since the inner product $\vec{v}\cdot \vec{v}=1$ in this case.

First, we 
give explicit assignments, with $\vec{v}$ only owning a even Hamming weight $w_2(\vec{v})=0$. There are totally eight ones, and they can be divided into three sets, each with four elements, i.e.,
\begin{equation}\label{Eq:CCZ3set}
\begin{aligned}
 &\{0000,0011,1100,1111\}, \{0000,0101,1010,1111\}, \\
 &\{0000,0110,1001,1111\}.
\end{aligned}
\end{equation}
It is easy to check that the inner product between the elements (also with itself) inside each set is $0$. As a result, one can randomly assign bit-string from any of three sets to the $m$ positions which satisfies the inner product constraint. This gives $3*4^m-2*2^m$, where $2^m$ accounts for the redundant counting of the case by only assigning bit-string from the set $\{0000,1111\}$.

Furthermore, we consider adding $\vec{v}$ with $w_2(\vec{v})=1$ to the previous assignments. For a given $\vec{v}$, as mentioned before, it can only appear in the $m$ position once. At the same time, the subsequent assignment of $\vec{v}'$ with $w_2(\vec{v}')=0$ is restricted. In particular, besides $\vec{v}'=0000$, one can only choose one element in each set in Eq.~\eqref{Eq:CCZ3set}, due to the inner product constraint. For example, suppose $\vec{v}=1000$, we can only choose $\vec{v}_1'=0011$, $\vec{v}_2'=0101$ and $\vec{v}_3'=0110$ besides $0000$. Note that one can not choose any two of them $\vec{v}_1', \vec{v}_2', \vec{v}_3'$ in the same time since the inner product between them is 1.
Consequently, the number of all possible assignments containing exact one element $\vec{v}$ with $w_2(\vec{v})=1$ is
$$C_8^1A_m^1*3*2^{m-1}=12m 2^m$$
where $C_8^1A_m^1$ represents choosing one $\vec{v}$ from total 8 odd-weight bit-strings and putting it in one of the $m$ positions; $3*2^{m-1}$ represents assigning bit-string to other $m-1$ positions from $\{0000,\vec{v}_i'\}$ for $i=1,2,3$.

We can continue this process by considering more than one odd-weight bit-string in the $m$ positions. In fact, this process can at most goes to four odd-weight bit-strings. To be specific, suppose one choose four distinct odd-weight strings, $\vec{v}_1,\vec{v}_2,\vec{v}_3,\vec{v}_4$, which are orthogonal to each other, they already form a complete basis of $\mbb{F}_2^4$. Consequently, there is no bit-string $v_5$ besides the trivial $0000$ such that the inner products of $v_5$ with $v_i,1\leq i\leq 4$ are all zero. For this case of four odd-weight strings, one can only assign $0000$ to other $m-4$ position. Totally $56A_m^4$.

Similar argument also holds for the case of three odd-weight strings, for example, $\vec{v}_1=1000,\vec{v}_2=0100,\vec{v}_3=0010$. One can not find a even-weight string $\vec{v}'$ orthogonal to all of them. Totally $8A_m^3$. For the case of two odd-weight strings, it depends on the specific choose. For example, if one select $\vec{v}_1=1000,\vec{v}_2=0111$, there are 3 legal even-weight strings besides $0000$, same as the one odd-weight string case; if $\vec{v}_1=1000,\vec{v}_2=0100$, there is only one nontrivial string $\vec{v}'=0011$. As a result, totally $(C_4^1*3+C_4^2*2)A_m^22^{m-2}$.

In summary, the number of all possible assignments is obtained by summing these numbers of all possible assignments

\comments{
\begin{equation}\label{}
\begin{aligned}
\#_{\mathrm{odd}}&<3(4^m-2^m)+\sum_{i=1}^3 3C_8^iA_m^i2^{(m-i)}+C_8^4\\
&<\sum_{i=1}^3 3\frac{8!}{(8-i)!i!}m^i2^{(m-i)}+C_8^4\\
&=3*4^m+3(7m^3+7m^2+4m-3)2^m+70\\
&<3*4^m+27m^32^m
\end{aligned}
\end{equation}

\begin{equation}\label{}
\begin{aligned}
\#_{\mathrm{odd}}&=12m2^m+(4+12*2^{m-2})A_m^2+8A_m^3+56A_m^4\\
&=(3m^2+9m)2^m+4A_m^2+8A_m^3+56A_m^4
\end{aligned}
\end{equation}
for $m\geq 4$.}
\begin{equation}\label{}
\begin{aligned}
\#_4&=3*4^m-2*2^m+12m2^m+24A_m^22^{m-2}+8A_m^3+56A_m^4\\
&=3*4^m+(6m^2+6m-2)2^m+4A_m^2+8A_m^3+56A_m^4\\
&=3*4^m+\Theta(m^2)2^m+\Theta(m^4)
\end{aligned}
\end{equation}
and finish the proof.

\end{proof}
\section{Proof of Corollary \ref{co:CCZVar} for $N_A$ qubits in subsystem $A$. }\label{ap:proofCoCCZ}
Here we prove Corollary \ref{co:CCZVar} for the general qubit number in $A$. As mentioned in main text, the variance of the purity for the full CCZ ensemble can be upper bounded by
\begin{equation}\label{}
\begin{aligned}
\delta^2_{CCZ}&< \langle P_A^2 \rangle_{CCZ,h}- \langle P_A \rangle_{CCZ}^2\\
&= \langle P_A^2 \rangle_{CCZ,h}-\langle P_A \rangle_{CCZ,h}^2+\langle P_A \rangle_{CCZ,h}^2- \langle P_A \rangle_{CCZ}^2\\
&=\delta^2_{CCZ,h}+\langle P_A \rangle_{CCZ,h}^2- \langle P_A \rangle_{CCZ}^2
\end{aligned}
\end{equation}
By inserting the results in Eq.~\eqref{eq:avrPurHCCZ}, \eqref{eq:avrPurCCZ} and \eqref{eq:VarHCCZ}, one has
\begin{equation}\label{}
\begin{aligned}
\delta^2_{CCZ}&< \frac{9d_{A}}{d^2}+3\frac{d_A+d_{\bar{A}}-1}{d}*\frac{d_A(d_A-1)N_{\bar{A}}({N_{\bar{A}}}+1)}{d^2}\\
&<\frac{9d_{A}}{d^2}+\frac{6d_{\bar{A}}}{d}*\frac{d_A^2N_{\bar{A}}({N_{\bar{A}}}+1)}{d^2}\\
&=\frac{d_{A}}{d^2}[9+6N_{\bar{A}}({N_{\bar{A}}}+1)]<3N^2 d^{-1.5}.
\end{aligned}
\end{equation}
by using the fact that $N_A\leq N_{\bar{A}}$ and $d_A\leq d_{\bar{A}}$.
\section{Proof of Th.~\ref{th:CCZVar}}\label{ap:proofCCZ}
Here we prove a tighter result for the variance of the purity of the full CCZ ensemble by counting the total surviving bit-strings. Similar as the average purity scenario in Sec.~\ref{CCZpurity}, the further introduction of the random gates $CCZ_{ijk}$ with $i,j\in A, k\in \bar{A}$ will induce more constraints on the bit configuration back from $\bar{A}$ to $A$. In particular, similar as Eq.~\eqref{eq:phMCCZvar} in Fig.~\ref{Fig:FlucAll} (c), the phase tensor follows directly by the tensor product of the one in Fig.~\ref{Fig:PurCCZ} (c), $\tilde{M}_{ijk}^{p}=\tilde{M}_{i_1j_1k_1}^p\otimes \tilde{M}_{i_2j_2k_2}^p$, and it equals $M_{kij}^{p}$ by Fig.~\ref{Fig:PurCCZ} (d).

In the following, we follows the cases \ref{halfCCZVarC1} to \ref{halfCCZVarC4} of the half CCZ ensemble in Sec.~\ref{CCZfluc}, and use the phase tensor to apply additional constraints from $\bar{A}$ to $A$.
Remembering in the discussion on the half CCZ ensemble in Sec.~\ref{CCZfluc}, only the case \ref{halfCCZVarC4} can contain bit-strings surviving in the formula of $\mbb{E}_{\Psi}(P_{A}^2)$ but not in the one of $[\mbb{E}_{\Psi}(P_{A})]^2$, which account for the final variance (see also the illustration in App.~\ref{ap:differ}). For completeness, here we list all the possible bit-stings, but finally we will only count the one that contributes to the net variance.

\begin{enumerate}[label=\textbf{\arabic*.},ref=\arabic*]
  \item $A$ only contains $\bar{00}$, $\bar{A}$ can be arbitrary, and vice versa. Totally $(4^{N_A}+4^{N_{\bar{A}}}-1)4^{N}$ with $4^N$ for the logical encoding.\label{ap:case1}
  \item $A$ only contains $\bar{00}$ and $\bar{01}$.
  The first two-bit of qubit in $\bar{A}$ can be arbitrary, and the last two should be restricted. We consider several possibilities for the last two-bit.
\begin{itemize}
  \item[\textbf{a}.] If they all take $\bar{0}$, it will not induce any constraint backwards on $A$, thus $(2^{N_A}-1)*(2^{N_{\bar{A}}}-1)4^{N}$.
  \item[\textbf{b}.] For the case there is one $01(10)$ or two of $01,10$, it means logical $\bar{1}$ on the last 2-bit.  As a result, it further induces constraint on the last 2-bit of the qubits in $A$, i.e.,  the logical $\bar{01}$ there. Similar as case \ref{halfCCZVarC2} in Sec.~\ref{CCZfluc}, one has
  totally $2^{N_{A}}N_A(N_A+1)*4^{N_{\bar{A}}}N_{\bar{A}}(N_{\bar{A}}+1)$.
\end{itemize}
Same result holds for $A$ only containing $\bar{00}$ and $\bar{10}$ by symmetry.
\item $A$ contains at least two types from $\{\bar{01},\bar{10},\bar{11}\}$. For the qubit in $\bar{A}$, both the first two and last two-bit are constrained independently. 
  Previously there are totally $4^{N_A}\left[4^{N_A}-3(2^{N_A}-1)-1\right]\left[2^{N_{\bar{A}}}+N_{\bar{A}}(N_{\bar{A}}+1)\right]^2$ bit-strings. The first term is for possibilities of $A$, and the second one is for $\bar{A}$. Since now different bit configurations induce different constraints backwards on $A$, we discuss the terms in $\left[2^{N_{\bar{A}}}+N_{\bar{A}}(N_{\bar{A}}+1)\right]^2$ separately.
  \begin{itemize}
  \item[\textbf{a}.] $2^{N_{\bar{A}}}*2^{N_{\bar{A}}}$. This term corresponds to that all bit-strings take $\bar{00}$ in $\bar{A}$, which has already been counted in case \ref{ap:case1}.
 \item[\textbf{a}.] $2*2^{N_{\bar{A}}}N_{\bar{A}}(N_{\bar{A}}+1)$. This term corresponds to the first 2-bits of all qubits in $\bar{A}$ take $\bar{0}$, and the last 2-bits take one $01(10)$ or two of $01,10$, or vice versa. For the first case, there is logical $\bar{1}$ on the last 2-bit and thus it induces constraint on the logical $\bar{01}$ and $\bar{11}$ in $A$. The number of the possibilities in $A$  reduces from $4^{N_A}\left[4^{N_A}-3(2^{N_A}-1)-1\right]$ to $[4(2^{N_A-1}-1)C_{N_A}^{1}+4(2^{N_A-2}-1)C_{N_A}^{2}+4(2^{N_A-2})C_{N_A}^{2}]*2^{N_A}\sim \Theta(N_A^2)4^{N_A}$. 
  \item[\textbf{c}.] $N_{\bar{A}}^2({N_{\bar{A}}}+1)^2$. This term corresponds to the case that both the first 2-bits and the last 2-bits of the qubits in $\bar{A}$ can take logical $\bar{1}$. 
  There are a few of different induced constraints on $A$ based on the specific bit configurations, and we will figure out the updated number of surviving bit-strings, denoted by $n_{3c}$, later.\label{ap:case3c}
\end{itemize}
  \item $A$ only contains $\bar{00}$ and $\bar{11}$. For the qubit in $\bar{A}$, both the first and the last two-bit are restricted jointly. The summation tensor is given in Eq.~\eqref{Eq:SumMFluCCZ4}, and previously there are $4^{N_A}(2^{N_A}-1)*\#_4(m=N_{\bar{A}})$ surviving bit-strings with $\#4 $ in Eq.~\eqref{Eq:L0011case4} by taking $m=N_{\bar{A}}$. Different choices of bit-strings from $\#4$, which is detailed in App.~\ref{proof:F24}, will induce different constraints backwards on $A$. We denote the updated number of surviving  bit-strings as $n_4$ and figure out it later.\label{ap:case4}
\end{enumerate}

To calculate the variance, one should in principle get $\mbb{E}_{\Psi}(P_{A}^2)$ by summing all the numbers of the above possibilities in cases \ref{ap:case1} to \ref{ap:case4}, and normalizing it with the constant in Eq.~\eqref{eq:NormalVar}. As discussed before, there are actually only a few bit-strings contributing to the final variance, since many of them in the above cases also survive in the formula of $[\mbb{E}_{\Psi}(P_{A})]^2$. In particular, it is not hard to check that only the above case \ref{ap:case3c} (c) and case \ref{ap:case4} can contain this kind of `genuine' bit-strings. As a result, one has the upper bound on the variance

\begin{equation}\label{ap:n3cn4}
\begin{aligned}
\delta^2_{CCZ}<(n_{3c}+n_4)d^{-4}.
\end{aligned}
\end{equation}

  Let us first figure out $n_4$ in case \ref{ap:case4}, by checking all the possible bit-strings in $\bar{A}$, described by $\#4$ shown explicitly in App.~\ref{proof:F24}. First, we are interested in the leading order in $\#_4$, that is, the configuration from the three even-hamming-weight sets shown in Eq.~\eqref{Eq:CCZ3set}. Note that the first set corresponds to all $\bar{00}$ case, which has been already counted in case \ref{ap:case1}. The last two sets both correspond to $\bar{00},\bar{11}$, and thus induce constraint backwards on $A$, which is described by the summation tensor in Eq.~\eqref{Eq:SumMFluCCZ4}, or more explicitly by Eq.~\eqref{eq:CCZvarc4} with $j,k\in A$ now. Just as the previous constraint from $A$ to $\bar{A}$, the number of legal bit configuration in $A$ is $2*4^{N_A}-2^{N_A}$, where we can only choose the last two sets in Eq.~\eqref{Eq:CCZ3set} now for $A$. Second, for the bit configuration with a few odd-hamming-weight bit-strings, there is $\bar{01}$ or $\bar{10}$ appearing in $\bar{A}$. For example, if there is $\bar{10}$ in $\bar{A}$, the first 2-bit on the logical $\bar{00}$ and $\bar{11}$ in $A$ will be restricted, and there are $2^{N_A}N_A(N_A+1)$ possibilities, and this number can further decrease to $\Theta(N_A^2)$ if there is also constraint on the last 2-bit. As a result, the number is case \ref{ap:case4} is bounded by
\begin{equation}\label{ap:n4}
\begin{aligned}
n_4<&(2*4^{N_A}-2^{N_A})(2*4^{N_{\bar{A}}}-2^{N_{\bar{A}}})+2^{N_A}N_A(N_A+1)*\Theta(N_{\bar{A}}^2)2^{N_{\bar{A}}}\\
=&4d^2-2(d_A+d_{\bar{A}})d+\Theta(N_A^2N_{\bar{A}}^2)d.
\end{aligned}
\end{equation}

Similarly, one can consider the bit-strings in case \ref{ap:case3c} (c). For previous total $N_{\bar{A}}^2({N_{\bar{A}}}+1)^2$ bit configurations, $2N_{\bar{A}}({N_{\bar{A}}}+1)$ of them correspond to $\bar{A}$ only containing $\bar{00},\bar{11}$. In this case, by following case \ref{ap:case4}, the number of possibilities on $A$ is upper bounded by the odd-hamming weight cases in $\#4(m=N_A)$, since the precondition is that $A$ has two types from $\{\bar{01},\bar{10},\bar{11}\}$. For the other cases, $\bar{A}$ contains two types from $\{\bar{01},\bar{10},\bar{11}\}$. One can directly follows case \ref{ap:case3c} (c), and the possibilities in $A$ is thus bounded by $N_{A}^2({N_{A}}+1)^2$. As a result, one has
\begin{equation}\label{ap:n3c}
\begin{aligned}
n_{3c}&=\Theta(N_A^2)2^{N_A}\Theta(N_{\bar{A}}^2)+\Theta(N_{A}^4)\Theta(N_{\bar{A}}^4)\\
&=\Theta(N_A^2N_{\bar{A}}^2)d_A.
\end{aligned}
\end{equation}

Inserting Eq.~\eqref{ap:n4} and \eqref{ap:n3c} into Eq.~\eqref{ap:n3cn4}, one has
\begin{equation}
\begin{aligned}
\delta^2_{CCZ}< 4d^{-2}-2(d_A+d_{\bar{A}})d^{-3}+\Theta(N_A^2N_{\bar{A}}^2)d^{-3}.
\end{aligned}
\end{equation}
In addition, by observing that the configurations in the first term of Eq.~\eqref{ap:n4}, that is, $(2*4^{N_A}-2^{N_A})(2*4^{N_{\bar{A}}}-2^{N_{\bar{A}}})+2^{N_A}N_A(N_A+1)$ are all genuine ones. Consequently, one also has
\begin{equation}
\begin{aligned}
\delta^2_{CCZ}> 4d^{-2}-2(d_A+d_{\bar{A}})d^{-3}.
\end{aligned}
\end{equation}
and we finish the proof. We remark that the more exact value of the variance is
\begin{equation}
\begin{aligned}
\delta^2_{CCZ}= 4d^{-2}-2(d_A+d_{\bar{A}})d^{-3}+\Theta(N_A^2N_{\bar{A}}^2)(d_A+d_{\bar{A}})d^{-4}.
\end{aligned}
\end{equation}
by counting the genuine configurations in $n_{3c}$ and $n_4$ more carefully. Since it shares the leading and sub-leading terms with Eq.~\eqref{eq:VarCCZ}, we do not elaborate it here.

\section{On the fluctuation of the entanglement entropy}\label{ap:entfluc}
\subsection{Proof of Prop.~\ref{prop:deviation} and discussion on the variance}\label{ap:entfluc1}
Here we first prove the deviation probability of Prop.~\ref{prop:deviation} in main text, and discuss the implication on the variance of the entanglement entropy. Hereafter for simplicity the $\log$ function is base 2 without clarification.
\begin{proof}
The probability of the entanglement entropy deviation can be bounded by Chebyshev inequality,
    \begin{equation}\label{ap:eqlowerS21}
\begin{aligned}
   \mathrm{Pr}\{-\log(P_A)\le-\log(\langle P_A \rangle+\varepsilon)\}
   =\mathrm{Pr}\{P_A-\langle P_A \rangle\ge\varepsilon\}
   \le \mathrm{Pr}\{|P_A-\langle P_A| \rangle\ge \varepsilon\} \le \frac{\delta^2(P_A)}{\varepsilon^2}.
   \end{aligned}
\end{equation}
The average purity $\langle P_A \rangle \le \frac{d_A+d_{\bar{A}}}{d}$ for both ensembles by Theorem \ref{th:CZPur} and \ref{th:CCZPur}, and one has
\begin{equation}\label{ap:eqlowerS2}
\begin{aligned}
-\log(\langle P_A \rangle+\varepsilon)\geq -\log\left( \frac{d_A+d_{\bar{A}}}{d}+\varepsilon\right)&= -\log\left( \frac{d_A+d_{\bar{A}}}{d}\right)-\log\left(1+\frac{\varepsilon d}{d_A+d_{\bar{A}}}\right)\\
&\geq -\log\left( \frac{d_A+d_{\bar{A}}}{d}\right)-\frac{\log(e)\varepsilon d}{d_A+d_{\bar{A}}}\geq -\log\left( \frac{d_A+d_{\bar{A}}}{d}\right)-1.5\varepsilon d_A.
   \end{aligned}
\end{equation}
where in the last-but-one inequality we use $\log_2(1+x)\leq \log_2(e) x$. Substituting $-\log(\langle P_A \rangle+\varepsilon)$ in Eq.~\eqref{ap:eqlowerS21} with the lower bound in Eq.~\eqref{ap:eqlowerS2}, we finish the proof.
\end{proof}
As mentioned in main text, we can use the deviation probability in Prop.~\ref{prop:deviation} to bound the variance of the entanglement entropy. Denote $\Delta=\frac{d_A+d_{\bar{A}}}{d}$, one bounds the variance by the expectation for the square of the difference to the maximal possible entropy $\log(d_A)=N_A$.
\begin{equation}\label{ap:eq:Varupper}
\begin{aligned}
\mathrm{Var}[S_2(\rho_A)]=&\mbb{E}_{\Psi} \left[S_2(\rho_A)-\langle S_2(\rho_A) \rangle \right]^2\\
\leq &\mbb{E}_{\Psi} \left[\log(d_A)-S_2(\rho_A)\right]^2\\
=&\mbb{E}_{\{\Psi|S_2(\rho_A)\leq-\log(\Delta)-\log(e)\varepsilon d_A\}} \left[\log(d_A)-S_2(\rho_A)\right]^2+\mbb{E}_{\{\Psi|S_2(\rho_A)>-\log(\Delta)-\log(e)\varepsilon d_A\}} \left[\log(d_A)-S_2(\rho_A)\right]^2\\
\leq &N_A^2\ \mathrm{Pr}\left\{S_2(\rho_A)\leq -\log(\Delta)-\log(e)\varepsilon d_A\right\}+\left[\log(d_A)+\log(\Delta)+\log(e)\varepsilon d_A\right]^2\\
\leq &N^2 \frac{\delta^2(P_A)}{\varepsilon^2}+\log^2(e)(d_A/d_{\bar{A}}+\varepsilon d_A)^2.
\end{aligned}
\end{equation}
In the last-but-one line, for the first term we note that $\left[\log(d_A)-S_2(\rho_A)\right]^2\leq N_A^2$, and later use Prop.~\ref{prop:deviation} to bound the probability for the small entropy cases; for the second term corresponding to the cases with entropy near $\log(d_A)$, we directly take the maximal deviation and the probability can be bounded trivially by $1$.

In the regime $d_A\ll d_{\bar{A}}$, one can choose an appropriate  $\delta(P_A) \ll \varepsilon \ll d_A^{-1}$ to control the variance exponentially small for both ensembles. To balance the two terms in Eq.~\eqref{ap:eq:Varupper}, one can choose $\varepsilon=\sqrt{\delta(P_A)d_A^{-1}}$. In this way, the entropy variance is in the order.
\begin{equation}
\begin{aligned}
\mathrm{Var}[S_2(\rho_A)]\sim \delta(P_A)d_A
\end{aligned}
\end{equation}
by omitting the insignificant coefficient of $N$. Remembering that the purity variances $\delta^2_{CZ}(P_A)=\Theta(d^{-1})$ and $\delta^2_{CCZ}(P_A)=\Theta(d^{-2})$ by Theorem \ref{th:CZVar} and Theorem \ref{th:CCZVar}, consequently the entropy variances for both ensembles can be quite small in the regime $d_A\ll d_{\bar{A}}$. It is not hard to see that this is also true for $S_1(\rho_A)$.

\subsection{Proof of Theorem \ref{th:CCZEntvar}}\label{ap:entfluc2}
Furthermore, in the regime that $d_A\sim d_{\bar{A}}\sim d^{1/2}$, the argument in Sec.~\ref{ap:entfluc1} can still apply for the CCZ ensemble. We prove Theorem \ref{th:CCZEntvar} as follows.
\begin{proof}
Similar as Eq.~\eqref{ap:eqlowerS21}, we can bound the probability for entropy deviation by Chebyshev inequality. Different from Eq.~\eqref{ap:eqlowerS21}, we now consider deviation on both sides.
\begin{equation}\label{ap:cczPrBoth}
\begin{aligned}
   \mathrm{Pr}\{-\log(P_A)\ge-\log(\langle P_A \rangle-\varepsilon)\}\bigcup\mathrm{Pr}\{-\log(P_A)\le-\log(\langle P_A \rangle+\varepsilon)\}
  \le \frac{\delta^2}{\varepsilon^2}.
   \end{aligned}
\end{equation}
By Theorem \ref{th:CCZPur} for the case of $d_A=d_{\bar{A}}=d^{1/2}$,
\begin{equation}\label{}
\begin{aligned}
-\log(\langle P_A \rangle-\varepsilon)< -\log\left( 2d^{-1/2}-d^{-1}-\varepsilon\right)&= -\log\left( 2d^{-1/2}\right)-\log\left[1-\frac1{2}(d^{-1}+\varepsilon)d^{1/2}\right]\\
&<(N/2-1)+\frac{1.1}{2}\log(e)(d^{-1}+\varepsilon)d^{1/2}\\
&<(N/2-1)+0.8(d^{-1}+\varepsilon)d^{1/2}
   \end{aligned}
\end{equation}
In the last-but-one line, we use the relation $-\log_2(1-x)=\log_2(e)(x+x^2/2+x^3/3\cdots)<1.1\log_2(e)x<0.8x$ for sufficient small $x=\frac1{2}(d^{-1}+\varepsilon)d^{1/2}$ here. By Eq.~\eqref{ap:eqlowerS2}, one also has
$-\log(\langle P_A \rangle+\varepsilon)\geq (N/2-1)-1.5\varepsilon d^{1/2}$. Since we take $d^{-1}\ll \varepsilon\ll d^{-\frac1{2}}$, one has $0.8(d^{-1}+\varepsilon)d^{1/2}<1.6\varepsilon d^{1/2}$. By inserting these bounds in Eq.~\eqref{ap:cczPrBoth} and note that $\delta^2_{CCZ}(P_A)<4d^{-2}$ in Theorem \ref{th:CCZVar}, one has
\begin{equation}
\begin{aligned}
\mathrm{Pr}\left\{\left|S_2(\rho_A)-\left(\frac{N}{2}-1\right)\right|\geq 1.6\varepsilon d^{1/2} \right\}\leq \frac{4d^{-2}}{\varepsilon^2}.\\
\end{aligned}
\end{equation}
The variance of the entanglement entropy can be proved in the same way as in Eq.~\eqref{ap:eq:Varupper}, but  consider the difference from $N/2-1$ here.
\begin{equation}\label{}
\begin{aligned}
\mathrm{Var}[S_2(\rho_A)]=&\mbb{E}_{\Psi} \left[S_2(\rho_A)-\langle S_2(\rho_A) \rangle \right]^2\\
\leq &\mbb{E}_{\Psi} \left[S_2(\rho_A)-(N/2-1)\right]^2\\
=&\mbb{E}_{\{\Psi|\left|S_2(\rho_A)-(N/2-1)\right|\geq 1.6\varepsilon d^{1/2}\}} \left[S_2(\rho_A)-(N/2-1)\right]^2+\mbb{E}_{\{\Psi|\left|S_2(\rho_A)-\left(\frac{N}{2}-1\right)\right|< 1.6\varepsilon d^{1/2}\}} \left[(N/2-1)-S_2(\rho_A)\right]^2\\
\leq &(N/2)^2\ \mathrm{Pr}\left\{\left|S_2(\rho_A)-\left(\frac{N}{2}-1\right)\right|\geq 1.6\varepsilon d^{1/2} \right\}+(1.6\varepsilon d^{1/2})^2\\
\leq &N^2\frac{d^{-2}}{\varepsilon^2}+1.6^2\varepsilon^2 d.
\end{aligned}
\end{equation}
By taking $\varepsilon=\sqrt{N/1.6}d^{-\frac{3}{4}}$ one has $\mathrm{Var}[S_2(\rho_A)]\leq 1.6Nd^{-\frac1{2}}$.
\end{proof}

\subsection{Proof of Theorem \ref{th:CZEntvar}}\label{ap:entfluc3}
For the CZ ensemble, as mentioned in main text, since the variance of the purity $\delta_{CZ}(P_A)=\Theta(d^{-1/2})$ is comparable to the expectation value of the purity $\langle P_A \rangle\sim d^{-1/2}$ in the regime $d_A\sim d_{\bar{A}}$, one expects a constant variance of entanglement entropy, and we prove Theorem \ref{th:CZEntvar} here.

Compared to studying the the variance $\mathrm{Var}[S_2(\rho_A)]$ with $\delta_{CZ}(P_A)$ in the previous subsections, here we directly evaluate $S(\rho_A)$ and then $\mathrm{Var}[S_2(\rho_A)]$ by the property of the graph state.

Firstly, we briefly review the formula of the entanglement entropy of graph state \cite{Hein2004Multiparty}.
Any simple graph $G$ can be uniquely determined by its adjacency matrix denoted as $\Gamma$, with $\Gamma_{i,j}=1$ iff $(i,j)\in E$. Suppose the vertex set $V=\{N\}$ is partitioned into two complementary subsets $A$ and $\bar{A}$, the adjacency matrix $\Gamma$ can be arranged in the following form,
\begin{equation}
 \Gamma_G=\left(
 \begin{array}{cc}
    \Gamma_A & \Gamma_{A\bar{A}}\\
    \Gamma_{A\bar{A}}^T & \Gamma_{\bar{A}}\\
  \end{array}
\right),
\end{equation}
where $\Gamma_A$, $\Gamma_{\bar{A}}$ describe the connections inside each subsystem, and the off-diagonal $N_A\times N_{\bar{A}}$ submatrix $\Gamma_{A\bar{A}}$ is for the ones between them.

Given a graph state $\ket{G}$ with its associated graph $G$, the reduced density matrix of a subsystem $A$ is $\rho_A=\mathrm{Tr}_{\bar{A}}(\ket{G}\bra{G})$, where the partial trace is on $\bar{A}$. The explicit formula of the entanglement entropy is
\begin{equation}\label{Eq:Entfor}
 S_2(\rho_A)=\mathrm{rank}(\Gamma_{A\bar{A}})
\end{equation}
where the $\mathrm{rank}$ is on the binary field $\mathbb{F}_2$. 
Note that Renyi-$\alpha$ entropy $S_\alpha(\rho_A)$ of any order is also suitable here, since the spectrum of $\rho_A$ is flat for graph states \cite{Hein2004Multiparty}.

The random graph state generated by random CZ gates just corresponds to randomly assigning $0/1$ with equal probability in the adjacency matrix $\Gamma$. If only considering the entanglement entropy, one only needs to study the statistical property of $\mathrm{rank}(\Gamma_{A\bar{A}})$ for random $\Gamma_{A\bar{A}}$.
The following Lemma considering the rank distribution of the random binary matrix \cite{kolchin_1998}.
\begin{lemma}\label{Le:Qs}\emph{(Th.~3.2.1 \cite{kolchin_1998})}
For the random $n\times n$ binary matrix $\Gamma$, with each element $\Gamma_{ij}$ distributes independently, and equally taking value $0/1$, the probability for the rank of $\Gamma$ on the binary field shows
\begin{equation}\label{}
\begin{aligned}
Q_s:=\mathrm{Pr}\{\mathrm{rank}(\Gamma)=n-s\}=2^{-s^2}\prod_{i\geq s+1}^{\infty}(1-2^{-i})\prod_{1\leq i\leq s}(1-2^{-i})^{-1}
\end{aligned}
\end{equation}
as $n\rightarrow \infty$. In particular, $Q_1=2Q_0$, $Q_2=\frac{4}{9}Q_0$, and numerically $Q_0\approx 0.288\cdots$.
\end{lemma}
Based on these preknowledges, we can prove Theorem \ref{th:CZEntvar} by considering the square matrix $\Gamma_{A\bar{A}}$ with $N_A=N_{\bar{A}}=N/2$ with $N\rightarrow \infty$.
\begin{proof}
From Corollary \ref{co:CZCCZEnt}, one knows that $\langle S_2(\rho_A) \rangle
\geq -\log\left(\frac{2\sqrt{d}}{d}\right)=N/2-1$. By using Lemma \ref{Le:Qs}, one has
\begin{equation}\label{}
\begin{aligned}
\mathrm{Var}[S_2(\rho_A)]&=\sum_{s=0}^{N/2}Q_s[(N/2-s)-\langle S_2(\rho_A) \rangle]^2\\
&>Q_2[(N/2-2)-\langle S_2(\rho_A) \rangle]^2\\
&>Q_2[(N/2-2)-(N/2-1)]^2=\frac{4}{9}Q_0\approx 0.128.
\end{aligned}
\end{equation}
\end{proof}
\end{document}